\documentclass[a4paper, pdftex, 12pt]{article} 

\usepackage{amsmath,amssymb,amsthm}
\usepackage{mathtools}
\usepackage{hyperref}
\usepackage{graphicx}
\usepackage{multirow}
\usepackage{here}
\usepackage{bm}
\usepackage{float}
\usepackage{layout}
\usepackage{latexsym}
\usepackage{subfigure}
\usepackage{booktabs}
\usepackage{lscape}
\usepackage{xcolor}
\hypersetup{
    colorlinks=true,
    citecolor=black,
    linkcolor=black,
    urlcolor=black
}
\usepackage{authblk}
\usepackage{colortbl}
\usepackage[top=30mm, bottom=30mm, left = 25mm, right = 25mm]{geometry}
\usepackage{natbib}
\usepackage{url}

\usepackage{xr}

\mathtoolsset{showonlyrefs}

\theoremstyle{definition}

\newtheorem{proposition}{Proposition}

\newtheorem{algorithm}{Algorithm}
\newtheorem{assumption}{Assumption}

\title{\textbf{Multi-resolution filters via linear projection \\for large spatio-temporal datasets}}

\author[1\footnote{E-mail: \href{mailto:1hirano2@kanto-gakuin.ac.jp}{1hirano2@kanto-gakuin.ac.jp}}]{Toshihiro Hirano}
\author[2]{Tsunehiro Ishihara}

\affil[1]{
Kanto Gakuin University}
\affil[2]{
  Takasaki City University of Economics}

\date{}

\begin{document}
\setlength{\baselineskip}{18.5pt}

\maketitle

\begin{abstract}
  Advances in compact sensing devices mounted on satellites have facilitated the collection of large spatio-temporal datasets with coordinates. Since such datasets are often incomplete and noisy, it is useful to create the prediction surface of a spatial field. 
  To this end, we consider an online filtering inference by using the Kalman filter based on linear Gaussian state-space models. However, the Kalman filter is impractically time-consuming when the number of locations in spatio-temporal datasets is large. To address this problem, we propose a multi-resolution filter via linear projection (MRF-lp), a fast computation method for online filtering inference. In the MRF-lp, by carrying out a multi-resolution approximation via linear projection (MRA-lp), 
  the forecast covariance matrix can be approximated while capturing both the large- and small-scale spatial variations. As a result of this approximation, our proposed MRF-lp preserves a block-sparse structure of some matrices appearing in the MRF-lp through time, which leads to the scalability of this algorithm. Additionally, we discuss extensions of the MRF-lp to a nonlinear and non-Gaussian case. Simulation studies and real data analysis for total precipitable water vapor demonstrate that our proposed approach performs well compared with the related methods.


\end{abstract}

\section{Introduction\label{sec_introduction}}

Recently, large-scale spatio-temporal datasets have been measured by compact sensing devices mounted on satellites. This type of remotely sensed data is often incomplete and noisy, and its statistical analysis, which allows for spatial and temporal dependence, would be beneficial in environmental science, traffic, and urban engineering. For example, creating a valid prediction surface for spatio-temporal datasets can be applied to forecasting an extreme weather phenomenon.

This study focuses primarily on a linear Gaussian state-space model given by observation and system model. 
Dynamical state-space models can flexibly model spatio-temporal datasets with missing values and/or noise and their spatio-temporal dynamics. For real-time or online filtering inference, we adopt the Kalman filter \citep{Kalman_1960} based on state-space models. Especially, the goal of this study is to obtain the filtering mean through the Kalman filter. However, the Kalman filter is computationally infeasible for spatio-temporal datasets with a large number of locations because it requires decomposing a large matrix at each time step.

The particle filter \citep[e.g.,][]{Gordon_1993, Kitagawa_1996} approximates the filtering and forecast distributions by some particles and is available for nonlinear and non-Gaussian state-space models. However, the particle filter often causes sample degeneracy when the dimension of the state vector is high \citep[e.g.,][]{Snyder_2008}. 
In the case of high dimension, the ensemble Kalman filter \citep[EnKF; e.g.,][]{Evensen_1994, Evensen_2009,Katzfuss_2016, Dong_2023}, which represents the filtering and forecast distributions by an ensemble, can be employed, but large ensemble sizes are often computationally infeasible in practice. 

In the case of large spatial datasets, there are numerous efficient statistical techniques. \cite{Heaton_2019} is a detailed review of these techniques. Furthermore, some techniques have recently been developed \citep[e.g.,][]{Katzfuss_2020d, Hirano_2021, Katzfuss_2021a, Peruzzi_2022}. \cite{Liu_2020} comprehensively reviews the state-of-the-art scalable Gaussian processes in machine learning literature. Some fast computation methods for large spatial datasets are extended to large spatio-temporal datasets. For instance, \cite{Datta_2016b}, \cite{Appel_2020} and \cite{Jurek_2021}, and \cite{Jurek_2022} are spatio-temporal versions of \cite{Datta_2016}, \cite{Katzfuss_2017}, and \cite{Katzfuss_2021a}, respectively. \cite{Peruzzi_2022} can also be applied to large spatio-temporal datasets. Many other computationally feasible methods for large spatio-temporal datasets have been proposed \citep[e.g.,][]{Wikle_1999, Wikle_2001, Johannesson_2007, Cressie_2010, Katzfuss_2011, Sigrist_2015, Bradley_2018, Elkhouly_2021, Chakraborty_2022}.

This paper proposes a multi-resolution filter via linear projection (MRF-lp) for large spatio-temporal datasets. This proposed method approximates the forecast covariance matrix in the Kalman filter by using a multi-resolution approximation via linear projection (MRA-lp) \citep[][]{Hirano_2021} which captures both the large- and small-scale spatial variations, unlike the covariance tapering \citep[e.g.,][]{Furrer_2006, Hirano_2013} and low rank approaches \citep[e.g.,][]{Banerjee_2008}. In addition, as a remarkable property, the MRF-lp preserves a block sparse structure of some matrices through time. 
Consequently, we can quickly calculate a good approximation of the filtering mean for large spatio-temporal datasets. 
Our proposed MRF-lp can be regarded as an extension of the MRA-lp to spatio-temporal datasets and a generalization of a multi-resolution filter (MRF) developed by \cite{Jurek_2021}. Additionally, unlike the MRF, the MRF-lp can empirically reduce ill-conditioning of certain matrices and increase numerical stability. 
The aforementioned desirable properties of the MRF-lp are inherited from the MRA-lp and MRF. 
Moreover, building upon the idea of \cite{Zilber_2021} and \cite{Jurek_2022}, we also discuss extensions of the MRF-lp to a nonlinear and non-Gaussian case where the system model is nonlinear, and the observation model is a distribution from an exponential family. The extensions are achieved by using the extended Kalman filter via the first-order Taylor expansion, the Laplace approximation, and the Newton--Raphson method. 

The rest of this paper is organized as follows. We review linear Gaussian state-space models and the Kalman filter for spatio-temporal datasets in Section \ref{sec_LGSSMs_and_Kalman_filter}. In Section \ref{sec_m_ra_lp_st_ssm}, we introduce the MRA-lp for spatio-temporal state-space models. Section \ref{sec_mrf_lp} describes our proposed MRF-lp. In Section \ref{sec_nonlinear_non-gaussian}, we discuss the extension to nonlinear and non-Gaussian state-space models. In Section \ref{sec_numerical_comparisons_illustrations}, we illustrate the results of the simulation studies and real data analysis. Finally, we conclude this study in Section \ref{sec_conclusion_future_studies}. The appendices contain the derivations of some algorithms, the technical proofs of the propositions, and the EnKF algorithm in the nonlinear and Gaussian case. In the supplementary material, we provide 
the derivations of the linear and nonlinear transformations in the system model and the detailed conditions in the simulation studies and real data analysis. 
To describe a physical evolution precisely and attain a sophisticated forecast, these transformations are derived from two kinds of partial differential equations in natural science.

\section{Linear Gaussian state-space models and the Kalman filter for spatio-temporal datasets\label{sec_LGSSMs_and_Kalman_filter}}

For $t = 1,\ldots, T$, we focus here on the following linear Gaussian state-space model
\begin{align}
 \bm{y}_t &= H_t \bm{x}_t + \bm{v}_t, \quad \bm{v}_t \sim \mathcal{N}_{n_t} \left( \bm{0}, R_t \right), \label{eq_observation_equation} \\
 \bm{x}_t &= A_t \bm{x}_{t-1} + \bm{w}_t, \quad \bm{w}_t \sim \mathcal{N}_{n_\mathcal{G}} \left( \bm{0}, Q_t \right), \label{eq_state_equation}
  \end{align}
where $\bm{x}_t= (x_{t1},\ldots,x_{t n_\mathcal{G}})^{\top}$ 
is an $n_\mathcal{G}$-dimensional state vector on a grid $\mathcal{G} = \{ \bm{g}_1, \ldots, \bm{g}_{n_\mathcal{G}} \}$ in a spatial domain $D_0$ at time $t$, and $\bm{x}_0 \sim \mathcal{N}_{n_\mathcal{G}} \left( \bm{\mu}_{0|0}, \Sigma_{0|0}  \right)$ is an initial state. $\bm{y}_t$ means an $n_t$-dimensional observation vector at time $t$ $(n_t \le n_\mathcal{G})$. 
$\bm{x}_0$, $\bm{v}_1, \ldots, \bm{v}_T, \bm{w}_1,\ldots,\bm{w}_T$ are independent, and $\Sigma_{0|0}$, $R_t$, and $Q_t$ 
are positive definite matrices. It is assumed that $R_t$, $Q_t$, $H_t$, $A_t$, $\bm{\mu}_{0|0}$, and $\Sigma_{0|0}$ are known. Moreover, based on \cite{Jurek_2021}, we assume that the transition matrix $A_t$ is sparse. This assumption is satisfied in several common cases, such as the discretization of the advection-diffusion equation (see Section \ref{suppl_advection_diffusion}) and the autoregressive (AR) process of order one on each grid point. 

In this study, we aim to obtain the mean vector of the filtering distribution $\bm{x}_t|\bm{y}_{1:t}$ $\left( \bm{y}_{1:t} = \left( \bm{y}_1^{\top},\ldots, \bm{y}_t^{\top} \right)^{\top}, \; t = 1,\ldots, T \right)$. 
By using the Kalman filter for \eqref{eq_observation_equation} and \eqref{eq_state_equation}, the filtering distributions for $t=1,\ldots, T$ are obtained as an online filtering inference. The Kalman filter starts from the forecast step with $\bm{\mu}_{0|0}$ and $\Sigma_{0|0}$ as the initial values and calculates the forecast distribution $p(\bm{x}_t|\bm{y}_{1:t-1}) = \mathcal{N}_{n_\mathcal{G}} \left( \bm{\mu}_{t|t-1}, \Sigma_{t|t-1} \right)$ and the filtering distribution 
$p(\bm{x}_t|\bm{y}_{1:t}) = \mathcal{N}_{n_\mathcal{G}} \left( \bm{\mu}_{t|t}, \Sigma_{t|t} \right)$ alternately.

\begin{algorithm}[Kalman filter]
  \label{alg_kf}
  Given $R_t$, $Q_t$, $H_t$, $A_t$, $\bm{\mu}_{0|0}$, and $\Sigma_{0|0}$, 
 find $\bm{\mu}_{t|t}$ and $\Sigma_{t|t}$ ($t=1,\ldots,T$). Set $t=1$.
  
  \bigskip
  \noindent
  \textit{Step} 1. Calculate $\bm{\mu}_{t|t-1} = A_t \bm{\mu}_{t-1|t-1}$ and $\Sigma_{t|t-1} = A_t \Sigma_{t-1|t-1} A_t^{\top} + Q_t$.
  
  \noindent
  \textit{Step} 2. Calculate $K_t = \Sigma_{t|t-1} H_t^{\top} \left( H_t \Sigma_{t|t-1} H_t^{\top} + R_t \right)^{-1}$. Next, we obtain $\bm{\mu}_{t|t} = \bm{\mu}_{t|t-1} + K_t \left( \bm{y}_t - H_t \bm{\mu}_{t|t-1} \right)$ and $\Sigma_{t|t} = \Sigma_{t|t-1} - K_t H_t \Sigma_{t|t-1}$. Set $t=t+1$. If $t \le T$, then go to Step 1. Otherwise, go to Step 3.
   
  \noindent
  \textit{Step} 3. Output $\bm{\mu}_{t|t}$ and $\Sigma_{t|t}$ ($t = 1,\ldots,T$).
  
  \end{algorithm}

In Algorithm \ref{alg_kf}, Steps 1 and 2 are the forecast and update steps, respectively. $K_t$ is called the Kalman gain. 
In the Kalman filter, $K_t$ includes the inverse matrix of the $n_t \times n_t$ matrix $H_t \Sigma_{t|t-1} H_t^{\top} + R_t$, and the operation count required for the inverse matrix calculation is proportional to the cube of its matrix size. Additionally, $K_t$ and $\Sigma_{t|t}$ require the product of matrices of sizes greater than or equal to $n_t$. Therefore, if the number of the observed locations, namely $n_t$, is large, calculating the filtering distribution in the Kalman filter at each time $t$ is computationally expensive. Since $A_t$ is the sparse matrix in this study, it is not burdensome to compute the forecast covariance matrix $\Sigma_{t|t-1}$.

\section{Multi-resolution approximation via linear projection for spatio-temporal state-space models
\label{sec_m_ra_lp_st_ssm}}

We approximate $\Sigma_{t|t-1}$ and $\Sigma_{0|0}$ by using the MRA-lp proposed by \cite{Hirano_2021}. This approximation leads to a fast computation algorithm of the Kalman filter in Section \ref{sec_mrf_lp}.

First, we define some notations in the MRA-lp in accordance with \cite{Hirano_2021}. Let $m$ ($m = 0,\ldots,M$) be a resolution. For $m = 0,\ldots,M$, $D_{j_1,\ldots,j_m}$ ($1 \le j_i \le J_i$, $2 \le J_i$, $i = 1,\ldots, M$) is a numbered subregion in $D_0$ at the $m$th resolution. Henceforth, for $m=0$, the index ($j_1 , \ldots, j_m$) and the index ($j_1, \ldots, j_m, a^{\prime}$) correspond to the index $0$ and the index $a^{\prime}$, respectively. Therefore, $D_{j_1,\ldots,j_m}$ for $m=0$ is $D_0$. In addition, $D_{j_1,\ldots,j_m}$ satisfies a domain partitioning rule
\begin{align*}
  D_{j_1,\ldots,j_m} = \bigcup_{j_{m+1} = 1, \ldots, J_{m+1}} D_{j_1,\ldots,j_m,j_{m+1}}, \quad 
  D_{j_1,\ldots,j_m,k} \cap D_{j_1,\ldots,j_m,l} = \emptyset,&\\
   \quad 1 \le k \neq l \le J_{m+1},&
\end{align*}
for $m=0,\ldots,M-1$. In the simulation studies and real data analysis, equal-area partitions are employed when the resolution increases. 
Let $I_{j_1,\ldots,j_m} = \{ i \mid \bm{g}_i \in D_{j_1,\ldots,j_m} \}$ denote a set of indices of grid points on $D_{j_1,\ldots,j_m}$. For $m=0$, $I_0 = \{ 1,\ldots,n_{\mathcal{G}} \}$. Suppose that the index $i$ of each grid point $\bm{g}_i$ is determined such that $\min \left( I_{j_1,\ldots,j_M} \right) > \max \left( I_{i_1,\ldots,i_M} \right)$ if $(j_1,\ldots,j_M) \succ_L (i_1,\ldots,i_M)$ where $\succ_L$ means lexicographic ordering. From the domain partitioning rule for $D_{j_1,\ldots,j_m}$, $I_{j_1,\ldots,j_{m+1}} \subset  I_{j_1,\ldots,j_m}$ ($m=0,\ldots,M-1$). 
Also, let $K_{j_1,\ldots,j_m}$ be a set of $r_{j_1,\ldots,j_m}$ indices of knots selected from $I_{j_1,\ldots,j_m}$ where $r_{j_1,\ldots,j_m} \le \left| I_{j_1,\ldots,j_m} \right|$ and $| \cdot |$ means the size of the set. We choose the knots randomly from the grid points on $D_{j_1,\ldots,j_m}$. Moreover, the grid points are not selected as knots more than once.

Finally, we define an $r^{\prime}_{j_1,\ldots, j_m} \times r_{j_1,\ldots, j_m}$ matrix $\Phi_{j_1,\ldots, j_m}$ ($m=0,\ldots,M$, $1 \le r^{\prime}_{j_1,\ldots, j_m} \le r_{j_1,\ldots, j_m}$) where $\operatorname{rank} \left( \Phi_{j_1,\ldots, j_m} \right) = r^{\prime}_{j_1,\ldots, j_m}$, and the row-norm of $\Phi_{j_1,\ldots, j_m}$ is equal to 1. To avoid the computational burden, it is posited that $r^{\prime}_{j_1,\ldots, j_m} \ll n_{\mathcal{G}}$. Intuitively, the role of $\Phi_{j_1,\ldots, j_m}$ is to generate $r^{\prime}_{j_1,\ldots,j_m}$ new compressed knots from $r_{j_1,\ldots,j_m}$ original knots. As shown in the simulation studies of \cite{Banerjee_2013} and \cite{Hirano_2021}, 
the numerical stability of the predictive process \citep{Banerjee_2008} and multi-resolution approximation (MRA) \citep{Katzfuss_2017} can be empirically improved by replacing the original knots with the compressed knots. Likewise, the MRF-lp with the compressed knots is numerically more stable than the MRF with the original knots (see the sixth simulation study in Section \ref{subsec_simulation_study} for details). Other merits of introducing $\Phi_{j_1,\ldots, j_m}$ are described, for example, in Section \ref{subsec_Overview_MRA-lp}. In Section \ref{subsec_select_phi}, we will discuss how to select $\Phi_{j_1,\ldots, j_m}$ in this study.

\subsection{Overview of the MRA-lp}
\label{subsec_Overview_MRA-lp}

In this section, based on \cite{Hirano_2021}, we will summarize the MRA-lp. To begin with, we introduce some notations that are needed except for those before Section \ref{subsec_Overview_MRA-lp}. However, note that the new additional notations are available only in Section \ref{subsec_Overview_MRA-lp} and are sometimes introduced with different definitions after the next subsection. Let $Y_0(\bm{s}) \sim \mbox{GP}(0, C_0)$ be a zero-mean Gaussian process with a covariance function $C_0(\bm{s},\bm{s}^{*})$ ($\bm{s}$, $\bm{s}^{*} \in D_0$
 ). Also, a set of knots on each subregion $D_{j_1,\ldots,j_m}$ is denoted by $Q_{j_1,\ldots, j_m}$ ($m=0,\ldots,M$). The knot means the location on $D_{j_1,\ldots,j_m}$. In addition, for a generic Gaussian process $Y(\bm{s}) \sim \mbox{GP}(0, C^{\prime})$ and sets of the vectors, that is, $A^{\prime} = \{ \bm{a}_1, \ldots,\bm{a}_N \}$ and $B^{\prime} = \{ \bm{b}_1, \ldots,\bm{b}_M \}$ ($\bm{a}_i, \bm{b}_j \in \mathbb{R}^{d'}$, $i= 1,\ldots,N,\; j=1,\ldots,M$), we write $\bm{Y}(A^{\prime}) = \left( Y(\bm{a}_1), \ldots, Y(\bm{a}_N) \right)^{\top}$ and $(C^{\prime}(A^{\prime},B^{\prime}))_{ij} = C^{\prime}(\bm{a}_i,\bm{b}_j)$ $(i= 1,\ldots,N,\; j=1,\ldots,M)$.

 First, we will explain the linear projection approach proposed by \cite{Banerjee_2013} for fast computation in large spatial datasets. For $\bm{s} \in D_0$, we define 
\begin{align*}
\tau_0(\bm{s}) &= E[Y_0(\bm{s}) | \Phi_0 \bm{Y}_0(Q_0)] \\
&=C_0(\bm{s},Q_0) \Phi_0^{\top} \left\{ \Phi_0  C_0(Q_0,Q_0) \Phi_0^{\top} \right\}^{-1} \Phi_0 \bm{Y}_0(Q_0).
\end{align*}
Then, it follows that
\begin{align*}
C_{\tau_0}(\bm{s}_1,\bm{s}_2) &= \mbox{Cov}(\tau_0(\bm{s}_1),\tau_0(\bm{s}_2))\\
&=C_0(\bm{s}_1,Q_0) \Phi_0^{\top} \widehat{V}_{0}^{0^{-1}} \Phi_0 C_0(\bm{s}_2,Q_0)^{\top},
\end{align*}
where $\widehat{V}_{0}^{0} = \Phi_0  V_0^0 \Phi_0^{\top}$ and $V_0^0 = C_0(Q_0,Q_0)$. The linear projection uses $C_{\tau_0}$ as an approximation of $C_0$ for fast computation in large spatial datasets and is identical with the predictive process proposed by \cite{Banerjee_2008} in the case of $\Phi_0 =\bm{\mbox{I}}_{\left| Q_0 \right|}$. 

In the MRA-lp, the calculation of $C_{\tau_0}$ is regarded as the linear projection at resolution 0. Based on \cite{Katzfuss_2017}, 
the MRA-lp implements the linear projection on $D_{j_1,\ldots,j_m}$ for the approximation error of $C_0$ at each $m$th resolution ($m=1,\ldots,M$). 
Specifically, for $\bm{s} \in D_0$, let $\delta_m(\bm{s}) \sim \mbox{GP}(0,C_m)$ ($m=1,\ldots,M$) be a zero-mean Gaussian process with the degenerate covariance function $C_m$ where
\begin{align*}
  C_m(\bm{s}_1,\bm{s}_2) =
  \begin{cases}
  C_{m-1}(\bm{s}_1,\bm{s}_2) - C_{\tau_{m-1}}(\bm{s}_1,\bm{s}_2), \quad \bm{s}_1,\bm{s}_2 \in D_{j_1,\ldots,j_m}\\
  \hspace{3.5cm} (1 \le j_i \le J_i,\; i=1,\ldots,m),\\
  0, \quad \mbox{otherwise}.
  \end{cases}
  \end{align*}
  By conducting the linear projection for $\delta_m(\bm{s})$ at the $m$th resolution, we obtain
  \begin{align*}
  \tau_m(\bm{s}) &= E[\delta_m(\bm{s}) | \Phi^{(m)} \bm{\delta}_m(Q^{(m)})]\\
  &=C_m(\bm{s},Q^{(m)}) \Phi^{{(m)}^{\top}} \left\{ \Phi^{(m)} C_m(Q^{(m)},Q^{(m)}) \Phi^{{(m)}^{\top}} \right\}^{-1} \Phi^{(m)}\bm{\delta}_m(Q^{(m)})
  \end{align*}
  and
  \begin{align*}
  C_{\tau_m}(\bm{s}_1,\bm{s}_2) &= \mbox{Cov}(\tau_m(\bm{s}_1),\tau_m(\bm{s}_2)) \notag \\
  &=
  \begin{cases}
  C_m(\bm{s}_1,Q_{j_1,\ldots,j_m}) \Phi_{j_1,\ldots,j_m}^{\top} 
  \widehat{V}_{j_1,\ldots,j_m}^{{m}^{-1}}
  \Phi_{j_1,\ldots,j_m} C_m(\bm{s}_2,Q_{j_1,\ldots,j_m})^{\top}, \\
  \hspace{3cm}  \bm{s}_1,\bm{s}_2 \in D_{j_1,\ldots,j_m} \; (1 \le j_i \le J_i,\; i=1,\ldots,m),\\
  0, \quad \mbox{otherwise},
  \end{cases}
  \end{align*}
  where $\widehat{V}_{j_1,\ldots,j_m}^m = \Phi_{j_1,\ldots,j_m} V_{j_1,\ldots,j_m}^m  \Phi_{j_1,\ldots,j_m}^{\top}$, $V_{j_1,\ldots,j_m}^m = C_m(Q_{j_1,\ldots,j_m},
  Q_{j_1,\ldots,j_m})$, \\
  $Q^{(m)}=\left\{ Q_{j_1,\ldots, j_m} \middle| 1 \le j_1 \le J_1, \ldots, 1 \le j_m \le J_m \right\}$, and 
  \begin{align*}
    \Phi^{(m)} = 
    \begin{pmatrix}
    \Phi_{1,\ldots,1} & & O \\
     & \ddots & \\
     O & & \Phi_{J_1,\ldots,J_m}
    \end{pmatrix}
    .
    \end{align*}
Finally, in the MRA-lp, the original covariance matrix $C_0(S_0,S_0)$ is approximated by $\sum_{m=0}^{M} C_{\tau_{m}}(S_0,S_0)$ where $S_0$ is a set of observed locations on $D_0$. Similar to the relationship between the linear projection and the predictive process, the MRA-lp is identical with the MRA in the case of $\Phi_{j_1,\ldots, j_m} = \bm{\mbox{I}}_{\left| Q_{j_1,\ldots, j_m} \right|}$.

In comparison with the MRA, the MRA-lp includes $\Phi_{j_1,\ldots,j_m}$ which compresses the knots by taking the linear combination. The introduction of $\Phi_{j_1,\ldots,j_m}$ has some advantages. First, we can separately select the number $r_{j_1,\ldots, j_m}$ of knots and the rank $r^{\prime}_{j_1,\ldots, j_m}$ of $\Phi_{j_1,\ldots,j_m}$. Therefore, we can employ more knots than the MRA by setting the low rank of $\Phi_{j_1,\ldots,j_m}$, which can lead to more efficient computation. Second, it may alleviate the selection problem of the knots. Finally, the remarkable property of $\Phi_{j_1,\ldots,j_m}$ is to yield the numerical stability in calculating the inverse matrix. 
The MRF-lp also has 
these merits obtained by introducing $\Phi_{j_1,\ldots,j_m}$ in contrast to the MRF.

\subsection{Algorithm for approximating the covariance matrix\label{subsec_mra_lp_alg}}

The MRA-lp is applied to $\Sigma_{t|t-1}$ and $\Sigma_{0|0}$. Here, we will describe the algorithm of the MRA-lp for approximating a generic $n_{\mathcal{G}} \times n_{\mathcal{G}}$ covariance matrix $\Sigma_0$. 
In what follows, for a generic $M^{'} \times M^{'}$ matrix $C$ and sets of indices, that is, 
$A^{\prime}, B^{\prime} \subseteq \{ 1,\ldots,M^{'} \}$, $C[A^{\prime},B^{\prime}]$ is a submatrix of $C$ obtained by restricting the rows and columns of $C$ to the indices in $A^{\prime}$ and $B^{\prime}$, respectively. 
The derivation of the following algorithm is given in Appendix \ref{subappend_derivation_alg1}.

\begin{algorithm}[Approximation of the generic covariance matrix $\Sigma_0$]
\label{alg_approx_cov_mat}

Given $M \ge 0$, $D_{j_1,\ldots,j_m}$ ($m = 0,\ldots, M$, $1 \le j_i \le J_i$, $i = 1,\ldots, M$), $I_{j_1,\ldots,j_m}$ ($m = 0,\ldots, M$, $1 \le j_i \le J_i$, $i = 1,\ldots, M$), and $K_{j_1,\ldots,j_m}$ ($m = 0,\ldots, M$, $1 \le j_i \le J_i$, $i = 1,\ldots, M$), find a matrix $B$ such that $\Sigma_0$ is approximated by $B B^{\top}$.

\bigskip
\noindent
\textit{Step} 1. Define $W_0^0 = \Sigma_0[I_0, K_0]$. By restricting the rows of $W_0^0$ to $K_0$, we obtain $V_0^0 = \Sigma_0[K_0, K_0]$. 
Next, calculate $\widehat{V}_{0}^{0}=\Phi_0 V_0^0 \Phi_0^{\top}$ and $B_0 = W_0^0 \Phi_0^{\top} \widehat{V}_0^{{0}^{-1/2}}$. If $M=0$, then go to Step 4. Set $m=0$.

\noindent
\textit{Step} 2. By restricting the rows of $W_{j_1,\ldots,j_m}^m$ to $I_{j_1,\ldots,j_{m+i}}$ ($i=1,\ldots, M-m$), we obtain $W_{j_1,\ldots,j_{m+i}}^m$ ($i=1,\ldots, M-m$). Similarly, $V_{j_1,\ldots,j_{m+i}}^m$ ($i=1,\ldots, M-m$) is given by restricting the rows of $W_{j_1,\ldots,j_{m+i}}^m$ ($i=1,\ldots, M-m$) to $K_{j_1,\ldots,j_{m+i}}$ ($i=1,\ldots, M-m$).

\noindent
\textit{Step} 3. Set $m=m+1$. For $1 \le j_i \le J_i$ ($i=1,\ldots,m$), calculate
\begin{align*}
W_{j_1,\ldots,j_m}^m = \Sigma_0[I_{j_1,\ldots,j_m}, K_{j_1,\ldots,j_m}]
- \sum_{k=0}^{m-1} W_{j_1,\ldots,j_m}^k \Phi_{j_1,\ldots,j_k}^{\top} \widehat{V}_{j_1,\ldots,j_k}^{{k}^{-1}} \Phi_{j_1,\ldots,j_k} V_{j_1,\ldots,j_m}^{{k}^{\top}},
\end{align*}
where $\widehat{V}_{j_1,\ldots,j_k}^{k} = \Phi_{j_1,\ldots,j_k} V_{j_1,\ldots,j_k}^k \Phi_{j_1,\ldots,j_k}^{\top}$ 
($k=0,\ldots,m-1$). 
We have $V_{j_1,\ldots,j_m}^m$ by restricting the rows of $W_{j_1,\ldots,j_m}^m$ to $K_{j_1,\ldots,j_m}$. Next, calculate $\widehat{V}_{j_1,\ldots,j_m}^{m} = \Phi_{j_1,\ldots,j_m} V_{j_1,\ldots,j_m}^m \Phi_{j_1,\ldots,j_m}^{\top}$ and $B_{j_1,\ldots,j_m} = W_{j_1,\ldots,j_m}^m \Phi_{j_1,\ldots,j_m}^{\top} \widehat{V}_{j_1,\ldots,j_m}^{{m}^{-1/2}}$. If $m < M$, then go to Step 2. Otherwise, go to Step 4.

\noindent
\textit{Step} 4. Output $B = \left( B^M \; B^{M-1} \; \cdots \; B^0 \right)$ where $B^m$ ($m = 0,\ldots, M$) is the block diagonal matrix whose diagonal element is $B_{j_1,\ldots,j_m}$ ($1 \le j_i \le J_i, \; i = 1, \ldots, m$), and the diagonal elements are arranged in ascending order with respect to lexicographic order of the index $(j_1,\ldots,j_m)$.

\end{algorithm}

The MRA-lp implements the linear projection \citep{Banerjee_2013} on each subregion obtained by partitioning the spatial domain recursively, resulting in the approximated covariance matrix that captures both the large- and small-scale spatial variations. \cite{Hirano_2021} demonstrates the effectiveness of the MRA-lp for the maximum likelihood estimation and the spatial prediction for large spatial datasets. To achieve the fast calculation of the Kalman filter, Algorithm \ref{alg_approx_cov_mat} omits a modification by using the covariance tapering at the highest resolution, unlike \cite{Hirano_2021}. This modification is helpful to enhance the approximation accuracy of spatial variations on small scales, but we can obtain an improvement similar to that of the covariance tapering by choosing large $M$. 

Algorithm \ref{alg_approx_cov_mat} does not require all elements of $\Sigma_0$ as inputs. The required matrices are $\Sigma_0[I_{j_1,\ldots,j_m}, K_{j_1,\ldots,j_m}]$ ($m=0,\ldots, M$, $1 \le j_i \le J_i$, $i = 1,\ldots, M$), which leads to the reduction of the computational cost and memory consumption of Algorithm \ref{alg_approx_cov_mat} in the fast calculation of the Kalman filter.

The approximated covariance matrix of $\Sigma_0$ by Algorithm \ref{alg_approx_cov_mat} is given by
\begin{align}
\label{eq_cov_mat_mra-lp}
B B^{\top} =& 
W_0^0 \Phi_0^{\top} \widehat{V}_{0}^{0^{-1}} \Phi_0 W_0^{{0}^{\top}} 
  +  
  \begin{pmatrix}
    W_1^1 \Phi_1^{\top} \widehat{V}_{1}^{1^{-1}} \Phi_1 W_1^{{1}^{\top}}  & & O \\
     & \ddots & \\
     O & & W_{J_1}^1 \Phi_{J_1}^{\top} \widehat{V}_{J_1}^{1^{-1}} \Phi_{J_1} W_{J_1}^{{1}^{\top}}
  \end{pmatrix} 
  \notag
\\
  + \cdots &+
  \begin{pmatrix}
    W_{1,\ldots,1}^M \Phi_{1,\ldots,1}^{\top} \widehat{V}_{1,\ldots,1}^{{M}^{-1}} \Phi_{1,\ldots,1} W_{1,\ldots,1}^{{M}^{\top}}  & & O \\
     & \ddots & \\
     O & & W_{J_1,\ldots,J_M}^{M} \Phi_{J_1,\ldots,J_M}^{\top} \widehat{V}_{J_1,\ldots,J_M}^{{M}^{-1}} \Phi_{J_1,\ldots,J_M} W_{J_1,\ldots,J_M}^{{M}^{\top}}
  \end{pmatrix}
.
\end{align}
The $(m+1)$st term ($m=0,\ldots,M$) in \eqref{eq_cov_mat_mra-lp} shows the linear projection at the $m$th resolution. It is observed that the linear projection at the higher resolution modifies the approximation of spatial variations on smaller and smaller scales.

Now, we provide the two propositions for Algorithm \ref{alg_approx_cov_mat}. 
The first proposition guarantees the existence of $\widehat{V}_{j_1,\ldots,j_m}^{{m}^{-1}}$ ($m=0,\ldots,M-1$) and $\widehat{V}_{j_1,\ldots,j_m}^{{m}^{-1/2}}$ ($m=0,\ldots,M$) in Algorithm \ref{alg_approx_cov_mat}.

\begin{proposition}
\label{prop_mra-lp_inv}
Suppose that we select $\Phi_{j_1,\ldots,j_m}$ ($m = 0, \ldots, M$), which satisfies $\mathrm{R}(\Phi_{j_1,\ldots,j_m}^{\top}) \cap \mathrm{R}(V_{j_1,\ldots,j_m}^m)^{\perp} = \{ \bm{0} \}$ ($m=1,\ldots,M$) if $M \ge 1$, where $\mathrm{R}(\cdot)$ and the symbol ``$\perp$'' mean the column space of $\cdot$ and the orthogonal complement, respectively. Then, $V_{j_1,\ldots,j_m}^m$ ($m=1,\ldots,M$) is the positive semidefinite matrix, and $\widehat{V}_{j_1,\ldots,j_m}^{m}$ ($m = 0, \ldots, M$) is the positive definite matrix.
\end{proposition}

The second proposition states a sparse structure of $B$ in Algorithm \ref{alg_approx_cov_mat}. This sparsity enables the fast computation of the Kalman filter. Now, we define $N=\sum_{m=0}^{M} r^{\prime}_{j_1,\ldots, j_m}$ for the fixed $(j_1,\ldots,j_M)$ and $N^{\prime} = \sum_{m=0}^M \sum_{1\le j_1 \le J_1,\ldots, 1 \le j_m \le J_m} r_{j_1,\ldots,j_m}^{\prime}$.

\begin{proposition}
  \label{prop_mra-lp_B}
$B$ is the $n_{\mathcal{G}} \times N^{\prime}$ matrix, and its row corresponding to the index in $I_{j_1,\ldots,j_M}$ has $N$ nonzero elements.
\end{proposition}

\subsection{Selection of $\Phi_{j_1,\ldots, j_m}$\label{subsec_select_phi}}

We will explain the selection of $\Phi_{j_1,\ldots, j_m}$ in Algorithm \ref{alg_approx_cov_mat}. In this paper, we select $\Phi_{j_1,\ldots,j_m}=U_{j_1,\ldots,j_m}^{{r_{j_1,\ldots,j_m}^{\prime}}^{\top}}$ where the $i$th column vector of $U_{j_1,\ldots,j_m}^{r_{j_1,\ldots,j_m}^{\prime}}$ is the eigenvector corresponding to the $i$th eigenvalue of $V_{j_1,\ldots,j_m}^m$ in descending order of magnitude for $i = 1, \ldots, r_{j_1,\ldots, j_m}^{\prime}$. If $1 \le r_{j_1,\ldots,j_m}^{\prime} < \mathrm{rank}\left( V_{j_1,\ldots,j_m}^m \right)$, then this selection of $\Phi_{j_1,\ldots,j_m}$ satisfies the assumption of Proposition \ref{prop_mra-lp_inv} from the argument of Section 3.2 of \cite{Hirano_2021}, and 
it follows that
\begin{align*}
\widehat{V}_{j_1,\ldots,j_m}^{m} &= 
\begin{pmatrix}
\lambda_{j_1,\ldots,j_m}^{(1)} & & O \\
& \ddots & \\
O & & \lambda_{j_1,\ldots,j_m}^{( r_{j_1,\ldots,j_m}^{\prime} )}
\end{pmatrix}, \; 
\widehat{V}_{j_1,\ldots,j_m}^{{m}^{-\frac{1}{2}}} &= 
\begin{pmatrix}
\dfrac{1}{\sqrt{\lambda_{j_1,\ldots,j_m}^{(1)}}} & & O \\
& \ddots & \\
O & & \dfrac{1}{\sqrt{\lambda_{j_1,\ldots,j_m}^{( r_{j_1,\ldots,j_m}^{\prime} )}}}
\end{pmatrix},
\end{align*}
where $\lambda_{j_1,\ldots,j_m}^{( i )}$ ($i = 1, \ldots, r_{j_1,\ldots, j_m}^{\prime}$) is the $i$th eigenvalue of $V_{j_1,\ldots,j_m}^m$ in descending order of magnitude. 

Although the calculation of the eigenvalues and eigenvectors of $V_{j_1,\ldots,j_m}^m$ involves $O(r_{j_1,\ldots,j_m}^3)$ computations \citep[e.g.,][]{Golub_2012}, this calculation in the simulation studies and real data analysis in Sections \ref{subsec_simulation_study} and \ref{subsec_real_data} gets computationally feasible by selecting adequate $r_{j_1,\ldots,j_m}$. 
If we need large $r_{j_1,\ldots,j_m}$ to achieve better approximation accuracy of the covariance matrix in Algorithm \ref{alg_approx_cov_mat}, we can take advantage of a stochastic matrix approximation technique on the basis of Algorithm 4.2 of \cite{Halko_2011} \citep[see Algorithm 2 of ][for details]{Hirano_2021}. By applying this technique to $V_{j_1,\ldots,j_m}^m$, the stochastic approximation of $\Phi_{j_1,\ldots,j_m}=U_{j_1,\ldots,j_m}^{{r_{j_1,\ldots,j_m}^{\prime}}^{\top}}$ is quickly obtained. However, whether the selected $\Phi_{j_1,\ldots,j_m}$ satisfies the assumption of Proposition \ref{prop_mra-lp_inv} rigorously is a future study. 
Simulation studies and real data analyses in \cite{Banerjee_2013}, \cite{Hirano_2017a}, and \cite{Hirano_2021} support the effectiveness of this technique.

\section{Multi-resolution filter via linear projection\label{sec_mrf_lp}}

To mitigate the computational burden of the Kalman filter for large spatio-temporal datasets, we will propose the MRF-lp for spatio-temporal state-space models and derive its properties and computational complexity.

\subsection{Fast calculation of the Kalman filter\label{subsec:fast_kalman_filter}}

The following algorithm allows us to efficiently calculate the filtering mean $\bm{\mu}_{t|t}$. We omit the derivation of Algorithm \ref{alg_mrf_lp} because it is the same as that in Section 3.3 of \cite{Jurek_2021}.

\begin{algorithm}[Multi-resolution filter via linear projection (MRF-lp)]
\label{alg_mrf_lp}

Given $M \ge 0$, $D_{j_1,\ldots,j_m}$ ($m = 0,\ldots, M$, $1 \le j_i \le J_i$, $i = 1,\ldots, M$), $I_{j_1,\ldots,j_m}$ ($m = 0,\ldots, M$, $1 \le j_i \le J_i$, $i = 1,\ldots, M$), and $K_{j_1,\ldots,j_m}$ ($m = 0,\ldots, M$, $1 \le j_i \le J_i$, $i = 1,\ldots, M$), find $\bm{\mu}_{t|t}$ ($t=1,\ldots,T$).

\bigskip
\noindent
\textit{Step} 1. Calculate $B_{0|0}$ by applying Algorithm \ref{alg_approx_cov_mat} to $\Sigma_{0|0}$. Set $t=1$.

\noindent
\textit{Step} 2. Calculate $\bm{\mu}_{t|t-1} = A_t \bm{\mu}_{t-1|t-1}$ and $B_{t|t-1}^F = A_t B_{t-1|t-1}$. Next, by applying Algorithm \ref{alg_approx_cov_mat} to $B_{t|t-1}^F B_{t|t-1}^{{F}^{\top}} + Q_t$, we obtain $B_{t|t-1}$.

\noindent
\textit{Step} 3. We have
\begin{align*}
\Lambda_t = \bm{\mathrm{I}}_{N^{\prime}} + B_{t|t-1}^{\top} H_t^{\top} R_t^{-1} H_t B_{t|t-1}.
\end{align*}
By computing the Cholesky decomposition of $\Lambda_t$, we obtain the lower triangular matrix $L_t$ as the Cholesky factor and calculate $B_{t|t} = B_{t|t-1} L_t^{{-1}^{\top}}$. Finally, calculate 
\begin{align*}
\bm{\mu}_{t|t} = \bm{\mu}_{t|t-1} + B_{t|t} B_{t|t}^{\top} H_t^{\top} R_t^{-1} \left( \bm{y}_t - H_t \bm{\mu}_{t|t-1} \right).
\end{align*}
Set $t=t+1$. If $t \le T$, then go to Step 2. Otherwise, go to Step 4.

\noindent
\textit{Step} 4. Output $\bm{\mu}_{t|t}$ ($t = 1,\ldots,T$).

\end{algorithm}

The MRF-lp provides the fast computation of $\bm{\mu}_{t|t}$ ($t = 1,\ldots,T$) by calculating $(\bm{\mu}_{t|t-1}, B_{t|t-1})$ and $(\bm{\mu}_{t|t}, B_{t|t})$ alternately. Steps 2 and 3 in Algorithm \ref{alg_mrf_lp} correspond to the forecast and update steps of the Kalman filter, respectively. In the MRF-lp, 
$\Sigma_{t|t-1}$ and $\Sigma_{t|t}$ are obtained by calculating $B_{t|t-1} B_{t|t-1}^{\top}$ and $B_{t|t} B_{t|t}^{\top}$, respectively. 
As mentioned in Section \ref{subsec_mra_lp_alg}, we do not have to completely calculate $B_{t|t-1}^F B_{t|t-1}^{{F}^{\top}} + Q_t$ in Step 2 because Algorithm \ref{alg_approx_cov_mat} requires just a part of it. 
Since $A_t$ is the sparse matrix, $\bm{\mu}_{t|t-1} = A_t \bm{\mu}_{t-1|t-1}$ and $B_{t|t-1}^F = A_t B_{t-1|t-1}$ in Step 2 
is computationally feasible.

The MRF-lp is the same as the MRF developed by \cite{Jurek_2021} except for Algorithm \ref{alg_approx_cov_mat} used in Algorithm \ref{alg_mrf_lp}. In addition, when $\Phi_{j_1,\ldots,j_m} = \bm{\mathrm{I}}_{r_{j_1,\ldots, j_m}}$ in Algorithm \ref{alg_approx_cov_mat}, the MRF-lp is 
identical with the MRF because Algorithm \ref{alg_approx_cov_mat} becomes the multi-resolution decomposition (MRD) 
in the MRF. This means that the MRF-lp can be regarded as a generalization of the MRF.

If the spatial data is not observed at $t=\tau, \tau+1,\ldots,\tau^{*}$, then $\bm{y}_{1:\tau-1} = \bm{y}_{1:\tau} = \cdots = \bm{y}_{1:\tau^{*}}$. In this case, at $t=\tau, \tau+1,\ldots,\tau^{*}$, the forecast step of the MRF-lp is the same as Step 2 of Algorithm \ref{alg_mrf_lp}, while the update step of the MRF-lp becomes $\bm{\mu}_{t|t}=\bm{\mu}_{t|t-1}$ and $B_{t|t}=B_{t|t-1}$. Likewise, for forecasting at $t=T+1, T+2,\ldots$, 
we calculate $\bm{\mu}_{t|T} = A_t \bm{\mu}_{t-1|T}$ and $B_{t|T}^F = A_t B_{t-1|T}$ and obtain $B_{t|T}$ by applying Algorithm \ref{alg_approx_cov_mat} to $B_{t|T}^F B_{t|T}^{{F}^{\top}} + Q_t$. If $\Sigma_{t|T}$ is unnecessary, we do not need to calculate $B_{t|T}$.

\subsection{Computational complexity\label{subsec_computational_complexity}}

In this section, we will elicit the properties of the MRF-lp and reveal that these properties lead to the fast computation of the MRF-lp for large spatio-temporal datasets. First, we impose some assumptions.

\begin{assumption}
  \label{assump_y}
  Elements of $\bm{y}_t$ are observed on some grid points in $\mathcal{G} = \{ \bm{g}_1, \ldots, \bm{g}_{n_\mathcal{G}} \}$ and arranged in ascending order by the index $i$ of $\bm{g}_i$. 
\end{assumption}

Under Assumption \ref{assump_y}, $H_t$ takes only one 1 in each row, and the rest are zeros. Assumption \ref{assump_y} removes the case where 
the elements of $\bm{y}_t$ are obtained by averaging some state vector elements.

\begin{assumption}
  \label{assump_y_one_observation}
  We obtain at least one observation in each $D_{j_1,\ldots,j_M}$.
\end{assumption}

\begin{assumption}
  \label{assump_R}
  $R_t[i,j]=0$ if the two grid points corresponding to $i$ and $j$ are not in the same subregion $D_{j_1,\ldots,j_M}$.
\end{assumption}

If $R_t$ is a diagonal matrix, Assumption \ref{assump_R} is satisfied. Additionally, we assume for simplicity that $r_{j_1,\ldots, j_m} = r_m$, $r^{\prime}_{j_1,\ldots, j_m} = r^{\prime}_{m}$, $J_m = J$, $n_\mathcal{G} = n$, and $n_t = c_t n$ ($0<c_t \le 1$). Note that this assumption is available only in Sections \ref{subsec_computational_complexity} and \ref{subsec_computational_complexity_nonlinear_nongaussian}. 
The following propositions show the block sparse structure of some matrices in the MRF-lp, which leads to the fast computation of the MRF-lp. 

\begin{proposition}
  \label{prop_property_sparse_mrf_lp}
Suppose that Assumptions \ref{assump_y}--\ref{assump_R} are satisfied. 

\par
\noindent
(a) $\Lambda_t$ is a block matrix with $M+1$ row blocks and $M+1$ column blocks. In addition, for $p,q = 0,\ldots, M$ ($p \ge q$), the size of the $(M+1-p, M+1-q)$th block is $J^p r_p^{\prime} \times J^q r_q^{\prime}$, and its block is a block diagonal matrix where each block has $r_q^{\prime}$ columns.
\par
\noindent
(b) We define the adjacency matrix as the matrix where the non-diagonal elements in the block parts of $\Lambda_t$ are 1, and the rest are 0. $G(V,E)$ denotes the undirected graph defined by this adjacency matrix where $V = \{ v_1,\ldots,v_{N^{\prime}} \}$, and $E$ is the set of edges. If $(v_i,v_j) \notin E$ for $1 \le j < i \le N^{\prime}$, 
$L_t[i,j] = L_t^{-1}[i,j] = 0$.
\par
\noindent
(c) In $B_{t|t}$ and $B_{t|t-1}$, the parts, which differ from those of the block diagonal elements $B_{j_1,\ldots,j_m}$ ($m=0,\ldots, M$) in the output matrix $B$ of Algorithm \ref{alg_approx_cov_mat}, are zero. 
\end{proposition}

Proposition \ref{prop_property_sparse_mrf_lp} (b) implies that the block sparse structures of the lower triangular parts of $L_t$ and $L_t^{-1}$ are the same as that of $\Lambda_t$. In Proposition \ref{prop_property_sparse_mrf_lp} (c), it is shown that $B_{t|t}$ has the same block sparse structure as $B_{t|t-1}$ in the MRF-lp, and this structure is preserved through time. 
From Propositions \ref{prop_property_sparse_mrf_lp} (a)--(c), we can significantly reduce the memory consumption of the MRF-lp as compared to that of the Kalman filter, where the calculation of $\Sigma_{t|t}$ and $\Sigma_{t|t-1}$ is needed. 
\begin{figure}[t]
   \subfigure[$\Lambda_t$]{
    \includegraphics[width = 5.1cm,pagebox=artbox,clip]{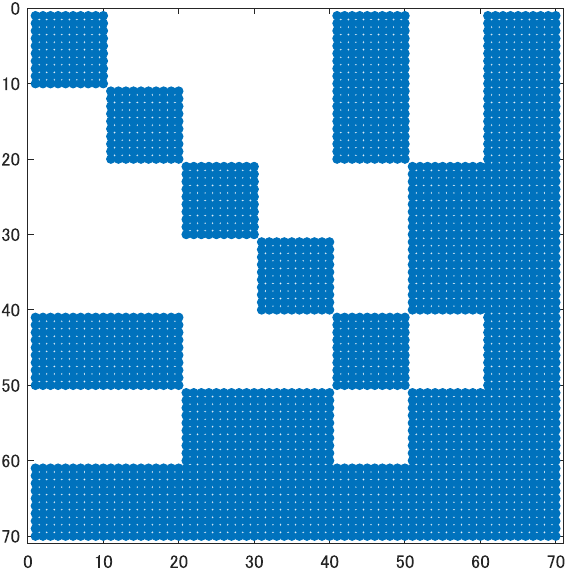}\label{fig:a}}
    \hspace{1cm}
   \subfigure[$L_t$ and $L_t^{-1}$]{ 
    \includegraphics[width = 5.1cm,pagebox=artbox,clip]{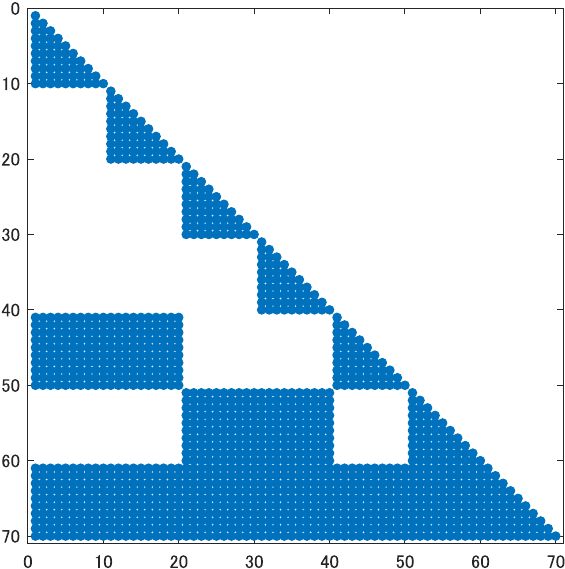}\label{fig:b}}
    \hspace{1cm}
   \subfigure[$B_{t|t}$ and $B_{t|t-1}$]{
    \includegraphics[width = 2.7cm,height=6cm,pagebox=artbox,clip]{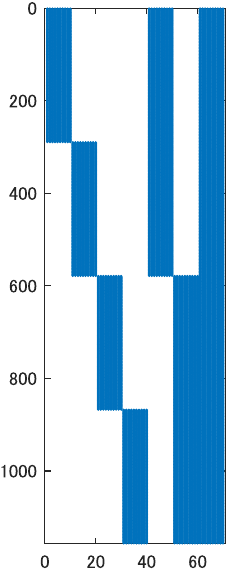}\label{fig:c}}
    \caption{
    Examples of $\Lambda_t$, $L_t$, $L_t^{-1}$, $B_{t|t}$, and $B_{t|t-1}$. Blue and white parts show nonzero elements and zero ones, respectively. We set $n=1156$,
    $M=2$, $J=2$, $r_m = 50$, and $r_m^{\prime} = 10$. Each domain is divided into equal subregions in Algorithm \ref{alg_approx_cov_mat}. 
    }
  \label{fig_example_proposition} 
  \end{figure}
Figure \ref{fig_example_proposition} depicts examples of $\Lambda_t$, $L_t$, $L_t^{-1}$, $B_{t|t}$, and $B_{t|t-1}$ in the MRF-lp. Each matrix shows the sparse structure stated in Proposition \ref{prop_property_sparse_mrf_lp}. 
We use the following proposition to assess the computational complexity of the MRF-lp.

\begin{proposition}
  \label{prop_property_L_L_inv}
  Under Assumptions \ref{assump_y}--\ref{assump_R}, the number of nonzero elements in each column of $L_t$ and $L_t^{-1}$ is $\mathcal{O}(N)$.
\end{proposition}

We will derive the theoretical order of the computational cost of the MRF-lp at each time step $t$. In the remaining part of this subsection, we further assume for simplicity that $r_m = r$, $r^{\prime}_{m} = r^{\prime}$, and $\left| I_{j_1,\ldots,j_M} \right| = \mathcal{O}(r)$. Thus, $N=(M+1)r^{\prime}$ and $N^{\prime} = \mathcal{O}(n)$. Furthermore, we require the following assumption which corresponds to Assumption 2 of \cite{Jurek_2021}.

\begin{assumption}
  \label{assump_A}
  $A_t$ has 
  $\mathcal{O}(r^{\prime})$ nonzero elements per row. Moreover, for a constant $c$, 
\begin{align}
  \left| \left\{ I_{j_1,\ldots,j_M} \mid \exists j \in I_{j_1,\ldots,j_M} \; \text{such that} \; A_t[i,j] \neq 0 \right\} \right| \le c, \quad i=1,\ldots,n, \; t=1,\ldots, T. \label{eq_assump_A}
\end{align}
\end{assumption}

Assumption \ref{assump_A} implies that the column numbers of $\mathcal{O}(r^{\prime})$ nonzero elements per row in $A_t$ belong to at most $c$ distinct $I_{j_1,\ldots,j_M}$. For example, in Sections \ref{subsec_simulation_study} and \ref{subsec_real_data}, $A_t$ derived from the discretization of the advection-diffusion equation satisfies \eqref{eq_assump_A} with $c=5$ (see Section \ref{suppl_advection_diffusion} for details).

In Step 2 of Algorithm \ref{alg_mrf_lp}, we can obtain $\bm{\mu}_{t|t-1} = A_t \bm{\mu}_{t-1|t-1}$ in $\mathcal{O}(n r^{\prime})$ from Assumption \ref{assump_A}. Since $B_{t|t-1}^F = A_t B_{t-1|t-1}$ has $\mathcal{O}(N)$ nonzero elements per row from Proposition \ref{prop_property_sparse_mrf_lp} (c) and Assumption \ref{assump_A}, $B_{t|t-1}^F$ is $\mathcal{O}(n r^{\prime} N)$ for the operation count. The input matrices to Algorithm \ref{alg_approx_cov_mat} can be computed in $\mathcal{O}(nNN^{\prime \prime})$ time where $N^{\prime \prime} = (M+1) r$. 
In Algorithm \ref{alg_approx_cov_mat}, from \cite{Golub_2012}, the eigenvalues and eigenvectors for all of $V_{j_1,\ldots,j_m}^m$ ($m = 0,\ldots, M$, $1 \le j_i \le J$, $i = 1,\ldots, M$) can be obtained in $\mathcal{O}(n r^{2})$, immediately leading to $\Phi_{j_1,\ldots,j_m}$, $\widehat{V}_{j_1,\ldots,j_m}^{m}$, and $\widehat{V}_{j_1,\ldots,j_m}^{{m}^{-1/2}}$ as shown in Section \ref{subsec_select_phi}. Since $\widehat{V}_{j_1,\ldots,j_m}^{m}$ and $\widehat{V}_{j_1,\ldots,j_m}^{{m}^{-1/2}}$ are the diagonal matrices, the total computational costs of $W_{j_1,\ldots,j_m}^m$ ($m=1,\ldots,M$, $1 \le j_i \le J$, $i = 1,\ldots, M$) and $B_{j_1,\ldots,j_m}$ ($m=0,\ldots,M$, $1 \le j_i \le J$, $i = 1,\ldots, M$) in Step 3 of Algorithm \ref{alg_approx_cov_mat} are $\mathcal{O}(n N^{{\prime \prime}^2})$ and $\mathcal{O}(n r^{\prime} N^{\prime \prime})$, respectively. 
The computational complexity of each procedure in Step 3 of Algorithm \ref{alg_mrf_lp} is derived based on the proof of Proposition 6 of \cite{Jurek_2021}. From a detailed calculation, computing $\Lambda_t$ requires $\mathcal{O}(n N^{{\prime\prime}^2})$ time. 
In addition, the argument of Sections 3.1 and 3.2 of \cite{Toledo_2021} shows that the computational complexity of the Cholesky decomposition is on the order of the sum of squares of the number of the nonzero elements in each column of $L_t$. Thus, from Proposition \ref{prop_property_L_L_inv}, the computational complexity of calculating $L_t$ is $\mathcal{O}(n N^2)$. 
Let $\bm{e}_i$ denote the $i$th column of the identity matrix. The $i$th column $\bm{a}_i$ of $L_t^{-1}$ is obtained by solving $L_t \bm{a}_i = \bm{e}_i$. Since $\bm{a}_i$ and each column of $L_t$ have $\mathcal{O}(N)$ nonzero elements from Proposition \ref{prop_property_L_L_inv}, the computational complexity of calculating $\bm{a}_i$ is $\mathcal{O}(N^2)$. Therefore, we have $L_t^{-1}$ in $\mathcal{O}(n N^2)$ time. For $B_{t|t} = B_{t|t-1} L_t^{{-1}^{\top}}$, its computational complexity is $\mathcal{O}(n N^2)$ because $B_{t|t}$ has $\mathcal{O}(n N)$ nonzero elements from Propositions \ref{prop_mra-lp_B} and \ref{prop_property_sparse_mrf_lp} (c), and each row of $B_{t|t-1}$ includes $\mathcal{O}(N)$ nonzero elements from Proposition \ref{prop_mra-lp_B}. 
As a consequence, Algorithm \ref{alg_mrf_lp} avoids $\mathcal{O}(n^3)$ operations at each time step $t$, unlike the Kalman filter.

\section{Extension to nonlinear and non-Gaussian case\label{sec_nonlinear_non-gaussian}}

Based on \cite{Zilber_2021} and \cite{Jurek_2022}, this section will discuss an extension of the MRF-lp to a nonlinear and non-Gaussian state-space model.

\subsection{Linear and non-Gaussian case\label{subsec_linear_non-gaussian}}

For $t = 1,\ldots, T$, we consider the following linear and non-Gaussian state-space model
\begin{align}
  y_{ti} | \bm{x}_t &\overset{ind.}{\sim} g_{ti}(y_{ti} \mid x_{ti}), \quad i \in I_{\bm{y}_t},   \notag
  \\ 
  \bm{x}_t | \bm{x}_{t-1} &\sim \mathcal{N}_{n_\mathcal{G}} \left( A_t \bm{x}_{t-1}, Q_t \right), \label{eq_state_model}
\end{align}
where $I_{\bm{y}_t}\subseteq I_0 = \{ 1,\ldots,n_{\mathcal{G}} \}$ ($n_t = \left| I_{\bm{y}_t} \right|$) is a set of indices of grid points on which each element $y_{ti}$ of $\bm{y}_t$ is observed and arranged in ascending order by the index $i$ of $\bm{g}_i$, meaning that Assumption \ref{assump_y} is already satisfied. $g_{ti}$ is a known distribution from the exponential family.

Our objective is to obtain $\bm{\mu}_{t|t}$ of an approximation $\hat{p}_L (\bm{x}_t|\bm{y}_{1:t}) = \mathcal{N}_{n_\mathcal{G}} \left( \bm{\mu}_{t|t}, \Sigma_{t|t} \right)$ of the filtering distribution $p(\bm{x}_t|\bm{y}_{1:t})$ by utilizing the Laplace approximation. When this approximation for the filtering distribution holds, the forecast distribution $p(\bm{x}_t|\bm{y}_{1:t-1})$ is also approximated by $\hat{p} (\bm{x}_t|\bm{y}_{1:t-1}) = \mathcal{N}_{n_\mathcal{G}} \left( \bm{\mu}_{t|t-1}, \Sigma_{t|t-1} \right)$. Since \eqref{eq_state_model} is linear, $\bm{\mu}_{t|t-1}$ and $\Sigma_{t|t-1}$ are calculated by the same procedure as the forecast step of the Kalman filter. For the filtering distribution, it follows that $p(\bm{x}_t|\bm{y}_{1:t}) \propto p(\bm{y}_t | \bm{x}_t) p(\bm{x}_t|\bm{y}_{1:t-1})$. By substituting $\hat{p} (\bm{x}_t|\bm{y}_{1:t-1})$ into $p(\bm{x}_t|\bm{y}_{1:t-1})$, we have $\hat{p}(\bm{x}_t|\bm{y}_{1:t}) \propto p(\bm{y}_t | \bm{x}_t) \hat{p}(\bm{x}_t|\bm{y}_{1:t-1}) = \prod_{i \in I_{\bm{y}_t}} g_{ti}(y_{ti}|x_{ti}) \mathcal{N}_{n_\mathcal{G}} \left( \bm{\mu}_{t|t-1}, \Sigma_{t|t-1} \right)$. By applying the Laplace approximation to $\hat{p}(\bm{x}_t|\bm{y}_{1:t})$, we will derive the approximated filtering distribution $\hat{p}_L (\bm{x}_t|\bm{y}_{1:t}) = \mathcal{N}_{n_\mathcal{G}} \left( \bm{\mu}_{t|t}, \Sigma_{t|t} \right)$.

For $i,j \in I_0$, by defining
\begin{align*}
  \left( \bm{u}_{\bm{x}_t} \right)_i &= 
  \begin{cases}
    \frac{d}{d x_{ti}} \log g_{ti}(y_{ti}|x_{ti}), \quad i \in I_{\bm{y}_t},\\
    0, \quad i \notin I_{\bm{y}_t},
  \end{cases}
  \\
  \left( D_{\bm{x}_t} \right)_{ij} &= 
  \begin{cases}
    -\frac{d^2}{d x_{ti}^2} \log g_{ti}(y_{ti}|x_{ti}), \quad i,j \in I_{\bm{y}_t} \; \mbox{and} \; i=j,\\
    0, \quad \mbox{otherwise},
  \end{cases}
\end{align*}
it follows that
\begin{align*}
  \frac{\partial}{\partial \bm{x}_t} \log \hat{p}(\bm{x}_t|\bm{y}_{1:t}) =& - \Sigma_{t|t-1}^{-1} (\bm{x}_t - \bm{\mu}_{t|t-1}) + \bm{u}_{\bm{x}_t}, \\
  -\frac{\partial^2}{\partial \bm{x}_t \partial \bm{x}_t^{\top}} \log \hat{p}(\bm{x}_t|\bm{y}_{1:t}) =& \Sigma_{t|t-1}^{-1} + D_{\bm{x}_t}\\
  =& W_{\bm{x}_t}, \quad (\text{say}).
\end{align*}
The diagonal elements of $D_{\bm{x}_t}$ are assumed to be greater than or equal to 0. Then, the Laplace approximation of $\hat{p}(\bm{x}_t|\bm{y}_{1:t})$ is
\begin{align*}
  \hat{p}_L (\bm{x}_t|\bm{y}_{1:t}) = \mathcal{N}_{n_\mathcal{G}} \left( \bm{\mu}_{t|t}, \Sigma_{t|t} \right),
\end{align*}
where $\bm{\mu}_{t|t} = \bm{\alpha}_t$, $\Sigma_{t|t} = W_{\bm{\alpha}_t}^{-1}$, and $\bm{\alpha}_t = \underset{\bm{x}_t \in \mathbb{R}^{n_\mathcal{G}}} {\operatorname{argmax}} \log \hat{p}(\bm{x}_t|\bm{y}_{1:t})$. By using the Newton--Raphson method, we can obtain $\bm{\alpha}_t$. The Newton--Raphson update is expressed as
\begin{align}
  h(\bm{x}_t) =& \bm{x}_t - \left( \frac{\partial^2}{\partial \bm{x}_t \partial \bm{x}_t^{\top}} \log \hat{p}(\bm{x}_t|\bm{y}_{1:t}) \right)^{-1} \left( \frac{\partial}{\partial \bm{x}_t} \log \hat{p}(\bm{x}_t|\bm{y}_{1:t}) \right) \notag \\
  =& \bm{x}_t + W_{\bm{x}_t}^{-1} \left\{ \Sigma_{t|t-1}^{-1} \left( \bm{\mu}_{t|t-1} - \bm{x}_t \right) + \bm{u}_{\bm{x}_t} \right\} \notag \\  
  =& \bm{\mu}_{t|t-1} + W_{\bm{x}_t}^{-1} \left\{ D_{\bm{x}_t} \left( \bm{x}_t - \bm{\mu}_{t|t-1} \right) + \bm{u}_{\bm{x}_t} \right\}. \label{eq_newton-raphson_update}
\end{align}
To obtain $\bm{\alpha}_t$, we start with an initial value $\bm{x}_{t}^{(0)}$ and update the current guess $\bm{x}_{t}^{(l)}$ ($l=0,1,2,\ldots$) iteratively until 
$\left\| \bm{x}_{t}^{(l+1)} - \bm{x}_{t}^{(l)} \right\|$ $\left(\bm{x}_{t}^{(l+1)} = h \left(\bm{x}_{t}^{(l)} \right)\right)$ 
is very small. However, \eqref{eq_newton-raphson_update} requires the calculation of $W_{\bm{x}_{t}^{(l)}}^{-1}$ at each iteration of the Newton--Raphson update, which causes the computational burden for large spatio-temporal datasets because its computational complexity is $O(n_\mathcal{G}^3)$. 

To resolve this problem, similar to Algorithm \ref{alg_mrf_lp}, we propose the MRF-lp for linear and non-Gaussian state-space models by applying Algorithm \ref{alg_approx_cov_mat} to
 $\Sigma_{t|t-1}$. The proposed algorithm enables us to calculate efficiently \eqref{eq_newton-raphson_update} (see Appendix \ref{subappend_mrf_lp_linear_non-gaussian} for the derivation of the algorithm).

\begin{algorithm}[Multi-resolution filter via linear projection (MRF-lp) for linear and non-Gaussian state-space models]
  \label{alg_mrf_lp_linear_non-gaussian}
  
  Given $M \ge 0$, $D_{j_1,\ldots,j_m}$ ($m = 0,\ldots, M$, $1 \le j_i \le J_i$, $i = 1,\ldots, M$), $I_{j_1,\ldots,j_m}$ ($m = 0,\ldots, M$, $1 \le j_i \le J_i$, $i = 1,\ldots, M$), $K_{j_1,\ldots,j_m}$ ($m = 0,\ldots, M$, $1 \le j_i \le J_i$, $i = 1,\ldots, M$), and $\varepsilon > 0$, find $\bm{\mu}_{t|t}$ ($t=1,\ldots,T$).
  
  \bigskip
  \noindent
  \textit{Step} 1. Calculate $B_{0|0}$ by applying Algorithm \ref{alg_approx_cov_mat} to $\Sigma_{0|0}$. Set $t=1$ and $l=0$.
  
  \noindent
  \textit{Step} 2. Calculate $\bm{\mu}_{t|t-1} = A_t \bm{\mu}_{t-1|t-1}$ and $B_{t|t-1}^F = A_t B_{t-1|t-1}$. Next, by applying Algorithm \ref{alg_approx_cov_mat} to $B_{t|t-1}^F B_{t|t-1}^{{F}^{\top}} + Q_t$, we obtain $B_{t|t-1}$.
  
  \noindent
  \textit{Step} 3. We have
  \begin{align*}
  \Lambda_{\bm{x}_{t}^{(l)}} = \bm{\mathrm{I}}_{N^{\prime}} + B_{t|t-1}^{\top} D_{\bm{x}_{t}^{(l)}} B_{t|t-1}.
  \end{align*}
  By computing the Cholesky decomposition of $\Lambda_{\bm{x}_{t}^{(l)}}$, we obtain the lower triangular matrix $L_{\bm{x}_{t}^{(l)}}$ as the Cholesky factor and calculate $B_{\bm{x}_{t}^{(l)}} = B_{t|t-1} L_{\bm{x}_{t}^{(l)}}^{{-1}^{\top}}$. Finally, calculate 
  \begin{align}
    \bm{x}_{t}^{(l+1)} = \bm{\mu}_{t|t-1} + B_{\bm{x}_{t}^{(l)}} B_{\bm{x}_{t}^{(l)}}^{\top} \left\{ D_{\bm{x}_{t}^{(l)}} \left( \bm{x}_{t}^{(l)} - \bm{\mu}_{t|t-1} \right) + \bm{u}_{\bm{x}_{t}^{(l)}} \right\}. \label{eq_fast_update_equation}
  \end{align}
  
  \noindent
  \textit{Step} 4. If $\| \bm{x}_{t}^{(l+1)} - \bm{x}_{t}^{(l)} \| \ge \varepsilon$, set $l=l+1$ and go to Step 3. Otherwise, 
  set $B_{t|t} = B_{\bm{x}_{t}^{(l+1)}}$, $\bm{\mu}_{t|t} = \bm{x}_{t}^{(l+1)}$, and $t = t+1$. In addition, if $t \le T$, set $l=0$ and go to Step 2. Otherwise, go to Step 5.

  \noindent
  \textit{Step} 5. Output $\bm{\mu}_{t|t}$ ($t = 1,\ldots,T$).
  
  \end{algorithm}

Algorithm \ref{alg_mrf_lp_linear_non-gaussian} avoids $\mathcal{O} \left(n_\mathcal{G}^3 \right)$ operations at each iteration in Step 3 as discussed in Section \ref{subsec_computational_complexity_nonlinear_nongaussian}. If $g_{ti}(y_{ti}|x_{ti}) = \mathcal{N} \left( x_{ti}, \tau^2 \right)$, Algorithm \ref{alg_mrf_lp_linear_non-gaussian} is compatible with Algorithm \ref{alg_mrf_lp} where $R_t = \tau^2 \bm{\mathrm{I}}_{n_t}$. 

\subsection{Nonlinear and non-Gaussian case\label{subsec_nonlinear_non-gaussian}}

Next, we extend Algorithm \ref{alg_mrf_lp_linear_non-gaussian} to the nonlinear and non-Gaussian case. For $t = 1,\ldots, T$, we focus on the following nonlinear and non-Gaussian state-space model
\begin{align}
  y_{ti} | \bm{x}_t &\overset{ind.}{\sim} g_{ti}(y_{ti} \mid x_{ti}), \quad i \in I_{\bm{y}_t},\\
  \bm{x}_t | \bm{x}_{t-1} &\sim \mathcal{N}_{n_\mathcal{G}} \left( \mathcal{A}_t (\bm{x}_{t-1}), Q_t \right), 
  \label{eq_system_model_nonlinear}
  \noeqref{eq_system_model_nonlinear}
\end{align}
  where $\mathcal{A}_t$ is a nonlinear differentiable function from $\mathbb{R}^{n_ \mathcal{G}}$ to $\mathbb{R}^{n_ \mathcal{G}}$.

Our goal is the same as the one in Section \ref{subsec_linear_non-gaussian}. Following \cite{Jurek_2022}, we extend the extended Kalman filter via the first-order Taylor expansion \citep[e.g.,][]{Durbin_2012} to the non-Gaussian case in the same way as the argument of Section \ref{subsec_linear_non-gaussian}, and then use Algorithm \ref{alg_approx_cov_mat} for $\Sigma_{t|t-1}$. The forecast distribution is assumed to follow the Gaussian distribution in the extended Kalman filter. Finally, we obtain an extension of Algorithm \ref{alg_mrf_lp_linear_non-gaussian} to the nonlinear case. 

\begin{algorithm}[Multi-resolution filter via linear projection (MRF-lp) for nonlinear and non-Gaussian state-space models]
  \label{alg_mrf_lp_nonlinear_non-gaussian}
  
Given $M \ge 0$, $D_{j_1,\ldots,j_m}$ ($m = 0,\ldots, M$, $1 \le j_i \le J_i$, $i = 1,\ldots, M$), $I_{j_1,\ldots,j_m}$ ($m = 0,\ldots, M$, $1 \le j_i \le J_i$, $i = 1,\ldots, M$), $K_{j_1,\ldots,j_m}$ ($m = 0,\ldots, M$, $1 \le j_i \le J_i$, $i = 1,\ldots, M$), and $\varepsilon > 0$, find $\bm{\mu}_{t|t}$ ($t=1,\ldots,T$).
  
  \bigskip
  \noindent
  \textit{Step} 1. Calculate $B_{0|0}$ by applying Algorithm \ref{alg_approx_cov_mat} to $\Sigma_{0|0}$. Set $t=1$ and $l=0$.
  
  \noindent
  \textit{Step} 2. Calculate the Jacobian matrix $A_t^{\prime} = \frac{\partial \mathcal{A}_t (\bm{x}_{t-1})}{\partial \bm{x}_{t-1}^{\top}} \bigl\vert_{\bm{x}_{t-1} = \bm{\mu}_{t-1|t-1}}$. We have $\bm{\mu}_{t|t-1} = \mathcal{A}_t \left( \bm{\mu}_{t-1|t-1} \right)$ and $B_{t|t-1}^F = A_t^{\prime} B_{t-1|t-1}$. Next, by applying Algorithm \ref{alg_approx_cov_mat} to $B_{t|t-1}^F B_{t|t-1}^{{F}^{\top}} + Q_t$, we obtain $B_{t|t-1}$.

  \noindent
  \textit{Steps} 3 and 4. The procedure is the same as that in Steps 3 and 4 in Algorithm \ref{alg_mrf_lp_linear_non-gaussian}.

  \noindent
  \textit{Step} 5. Output $\bm{\mu}_{t|t}$ ($t = 1,\ldots,T$).
  
  \end{algorithm}

  If $\mathcal{A}_t (\bm{x}_{t-1}) = A_t \bm{x}_{t-1}$, then Algorithm \ref{alg_mrf_lp_nonlinear_non-gaussian} is identical with Algorithm \ref{alg_mrf_lp_linear_non-gaussian}. Furthermore, 
  Algorithm \ref{alg_mrf_lp_nonlinear_non-gaussian} with $\Phi_{j_1,\ldots,j_m} = \bm{\mathrm{I}}_{r_{j_1,\ldots, j_m}}$ can be regarded as an extension of the MRF of \cite{Jurek_2021} to the nonlinear and non-Gaussian case.

\subsection{Computational complexity in nonlinear and non-Gaussian case
\label{subsec_computational_complexity_nonlinear_nongaussian}}

As in Section \ref{subsec_computational_complexity}, it is assumed that $r_{j_1,\ldots, j_m} = r_m$, $r^{\prime}_{j_1,\ldots, j_m} = r^{\prime}_{m}$, $J_m = J$, $n_\mathcal{G} = n$, and $n_t = c_t n$ ($0<c_t \le 1$) only in Sections \ref{subsec_computational_complexity} and \ref{subsec_computational_complexity_nonlinear_nongaussian}. 

\begin{proposition}
  \label{prop_property_sparse_mrf_lp_non-Gaussian}
Suppose that Assumption \ref{assump_y_one_observation} is satisfied. Even if we replace $\Lambda_t$, $L_t$, $L_t^{-1}$, and $B_{t|t}$ in Propositions \ref{prop_property_sparse_mrf_lp} (a)--(c) and \ref{prop_property_L_L_inv} by $\Lambda_{\bm{x}_{t}^{(l)}}$, $L_{\bm{x}_{t}^{(l)}}$, $L_{\bm{x}_{t}^{(l)}}^{-1}$, and $B_{\bm{x}_{t}^{(l)}}$, respectively, their statements hold, aside from the assumptions in Propositions \ref{prop_property_sparse_mrf_lp} (a)--(c) and \ref{prop_property_L_L_inv}.

\end{proposition}

Moreover, assume that $r_m = r$, $r^{\prime}_{m} = r^{\prime}$, and $\left| I_{j_1,\ldots,j_M} \right| = \mathcal{O}(r)$ only in the rest of this subsection. Proposition \ref{prop_property_sparse_mrf_lp_non-Gaussian} implies that $\Lambda_{\bm{x}_{t}^{(l)}}$, $L_{\bm{x}_{t}^{(l)}}$, $L_{\bm{x}_{t}^{(l)}}^{-1}$, $B_{\bm{x}_{t}^{(l)}}$, and $B_{t|t-1}$ in Algorithms \ref{alg_mrf_lp_linear_non-gaussian} and \ref{alg_mrf_lp_nonlinear_non-gaussian} have the same sparse structures as the corresponding matrices in Algorithm \ref{alg_mrf_lp} for all $t$ and $l$. Therefore, by an argument similar to Section \ref{subsec_computational_complexity}, the computational costs of $\Lambda_{\bm{x}_{t}^{(l)}}$, $L_{\bm{x}_{t}^{(l)}}$, $L_{\bm{x}_{t}^{(l)}}^{-1}$, and $B_{\bm{x}_{t}^{(l)}}$ for each $t$ and $l$ are $\mathcal{O}(n N N^{\prime\prime})$, $\mathcal{O}(n N^2)$, $\mathcal{O}(n N^2)$, and $\mathcal{O}(n N^2)$, respectively. In Step 2 of Algorithm \ref{alg_mrf_lp_linear_non-gaussian}, 
we can avoid $\mathcal{O}(n^3)$ operations under Assumption \ref{assump_A}. Consequently, Algorithm \ref{alg_mrf_lp_linear_non-gaussian} does not include $\mathcal{O}(n^3)$ operations such as $W_{\bm{x}_{t}^{(l)}}^{-1}$ at each $t$ and $l$. 

In the nonlinear and non-Gaussian case, Algorithm \ref{alg_mrf_lp_nonlinear_non-gaussian} is the same as Algorithm \ref{alg_mrf_lp_linear_non-gaussian} except for Step 2. In Step 2 of Algorithm \ref{alg_mrf_lp_nonlinear_non-gaussian}, we need to calculate the nonlinear transformation $\mathcal{A}_t \left( \bm{\mu}_{t-1|t-1} \right)$ and its Jacobian matrix $A_t^{\prime}$, and their computational costs depend on the problem setting. Additionally, $A_t^{\prime}$ is not necessarily a sparse matrix. However, in our simulation studies of Section \ref{subsec_simulation_study}, we obtain $\mathcal{A}_t$ and 
sparse $A_t^{\prime}$ by applying the fourth-order Runge-Kutta method to a system of partial differential equations in \cite{Lorenz_2005}, and 
Algorithm \ref{alg_mrf_lp_nonlinear_non-gaussian} decreases the computation time to some extent in comparison with 
the original method without Algorithm \ref{alg_approx_cov_mat}.

\section{Numerical comparisons and illustrations\label{sec_numerical_comparisons_illustrations}}

In this section, we compared our proposed MRF-lp with the MRF and the original methods, including the Kalman filter, by using the simulated and real data based on \cite{Jurek_2021, Jurek_2022, Jurek_2023}. All computations were done by using MATLAB on a single-core machine (4.20 GHz) with 64 GB RAM. 

\subsection{Simulation study\label{subsec_simulation_study}}

Through simulation studies, we evaluated the performances of Algorithms \ref{alg_mrf_lp}--\ref{alg_mrf_lp_nonlinear_non-gaussian}. We used Algorithms \ref{alg_mrf_lp}--\ref{alg_mrf_lp_nonlinear_non-gaussian} with $\Phi_{j_1,\ldots,j_m} = \bm{\mathrm{I}}_{r_{j_1,\ldots, j_m}}$ as the MRF, meaning $r_{j_1,\ldots, j_m} = r_{j_1,\ldots, j_m}^{\prime}$ in the MRF. The original methods in the non-Gaussian case correspond to Algorithms \ref{alg_mrf_lp_linear_non-gaussian} and \ref{alg_mrf_lp_nonlinear_non-gaussian} where Algorithm \ref{alg_approx_cov_mat} is not applied. 
Additionally, we conducted a comparison with the EnKF.

Let $D_0=[0, 1]^2$, $T=20$, and $\bm{\mu}_{0|0} = \bm{0}$. $\bm{x}_t$ was on the $n_\mathcal{G}=34 \times 34=1156$ grid points. The missing locations were selected randomly over $D_0$ at each time, and we set $n_t=0.3 n_\mathcal{G}$ and $R_t = 0.05 \bm{\mathrm{I}}_{n_t}$. $Q_t$ and $\Sigma_{0|0}$ were based on the exponential covariance function with the range parameter $1/0.15$, and their variances 
were 0.1 and 1, respectively. This situation means that most variation stems from the model, but the model error is also unignorable. For the MRF, we considered two cases: $M=2$, $r_{j_1,\ldots,j_m} = 10$, $J_m=2$ and $M=4$, $(r_0,r_{j_1},r_{j_1,j_2},r_{j_1,j_2,j_3},r_{j_1,j_2,j_3,j_4}) = (10,10,10,5,5)$, $J_m=2$. Similarly, for the MRF-lp, we adopted these two cases except for $r_{j_1,\ldots,j_m} = 50$, $r_{j_1,\ldots,j_m}^{\prime} = 10$ in $M=2$ and $(r_0,r_{j_1},r_{j_1,j_2},r_{j_1,j_2,j_3},r_{j_1,j_2,j_3,j_4}) = (50,50,50,10,10)$, $(r_0^{\prime},r_{j_1}^{\prime},r_{j_1,j_2}^{\prime},r_{j_1,j_2,j_3}^{\prime},r_{j_1,j_2,j_3,j_4}^{\prime}) = (10,10,10,5,5)$ in $M=4$. 
We employ the conditions in this paragraph as the baseline in Section \ref{subsec_simulation_study} and state partial changes, if any.

All comparisons were conducted based on the mean squared prediction error (MSPE) of $\bm{\mu}_{t|t}$ for $\bm{x}_t$. At each time point, we calculated the averaged MSPE of each method from the ten iterations and recorded the ratios of the averaged MSPEs of the MRF and MRF-lp relative to that of the original method. Furthermore, we also averaged their averaged MSPEs over time and scaled them relative to that of the original method. Likewise, the average value of the total computational time for the ten iterations was recorded and scaled relative to that of the original method.

First, we assessed the performance of Algorithm \ref{alg_mrf_lp} in the linear and Gaussian case. Only for this case, 
we considered additional scenarios where a part of the conditions in the baseline was changed: 
the small sample case where $n_t=0.1 n_\mathcal{G}$, the low noise case where $R_t = 0.02 \bm{\mathrm{I}}_{n_t}$, and the smooth case where $Q_t$ and $\Sigma_{0|0}$ were based on the Mat\'ern covariance function with the smoothness parameter 1.5. Following \cite{Xu_2007}, \cite{Stroud_2010}, and \cite{Jurek_2021}, $A_t$ was derived from the discretization of the advection-diffusion equation by using the first-order forward differences in time and the centered differences in space (see Section \ref{suppl_advection_diffusion} for the derivation and conditions of $A_t$). As a consequence, $A_t$ becomes a sparse matrix. 

\begin{figure}[p]
  \centering
  \begin{minipage}[t]{6cm}
 \subfigure[Baseline\\ (MRF: 2.513, MRF-lp: 1.927)]{
  \includegraphics[width = 5.2cm,pagebox=artbox,clip]{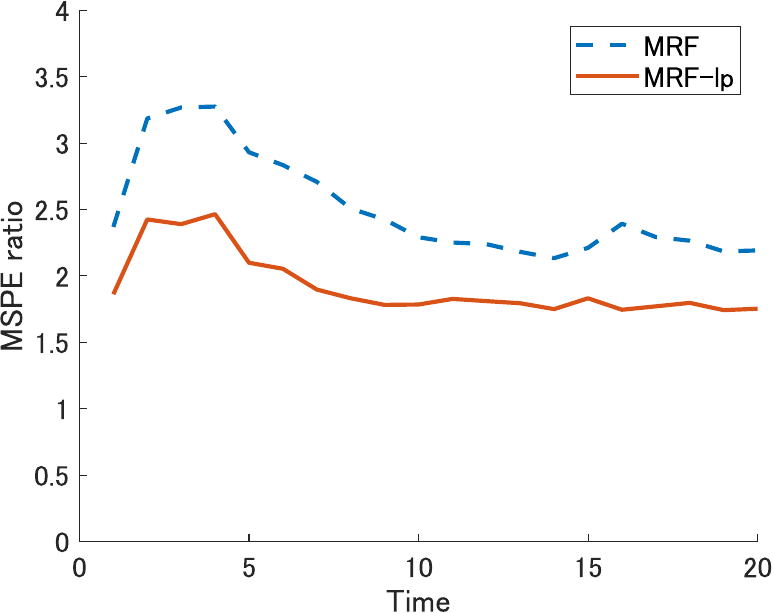}}\\
 \subfigure[Small sample case\\ (MRF: 1.602, MRF-lp: 1.356)]{ 
  \includegraphics[width = 5.2cm,pagebox=artbox,clip]{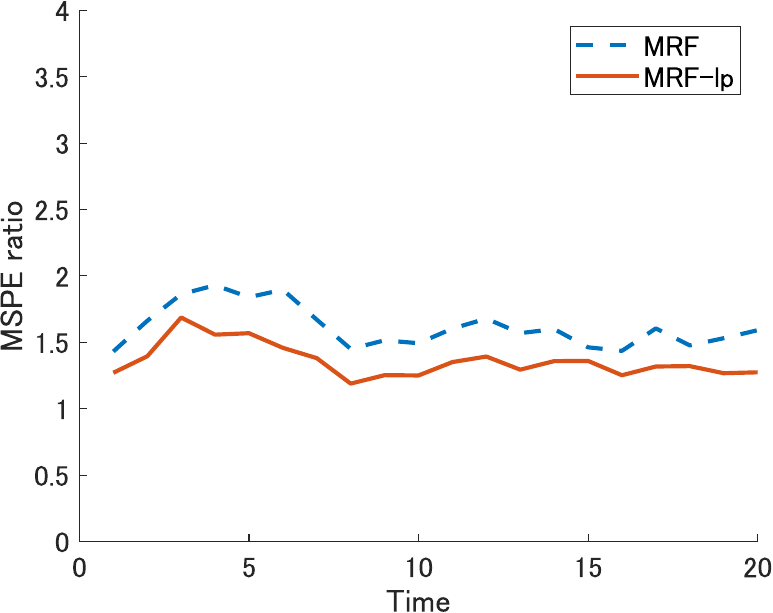}}\\
 \subfigure[Low noise case\\ (MRF: 2.893, MRF-lp: 2.278)]{
  \includegraphics[width = 5.2cm,pagebox=artbox,clip]{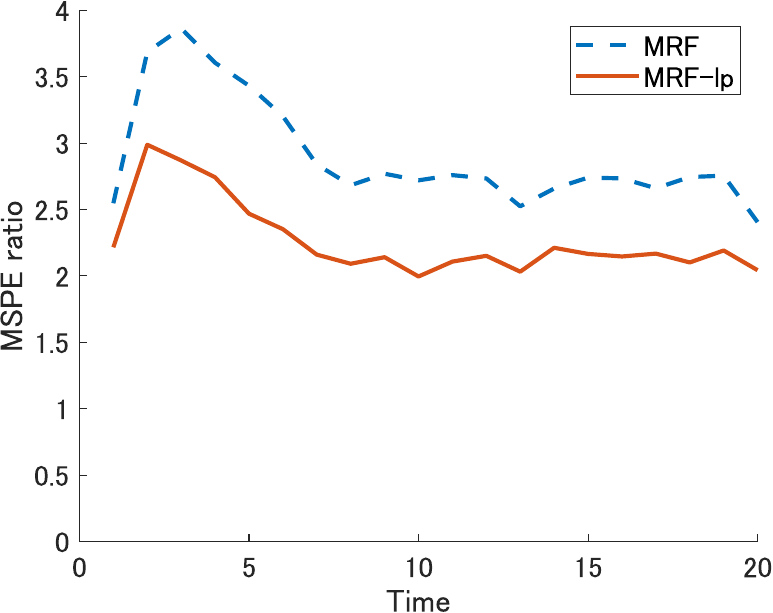}}\\
 \subfigure[Smooth case\\ (MRF: 1.682, MRF-lp: 1.356)]{
  \includegraphics[width = 5.2cm,pagebox=artbox,clip]{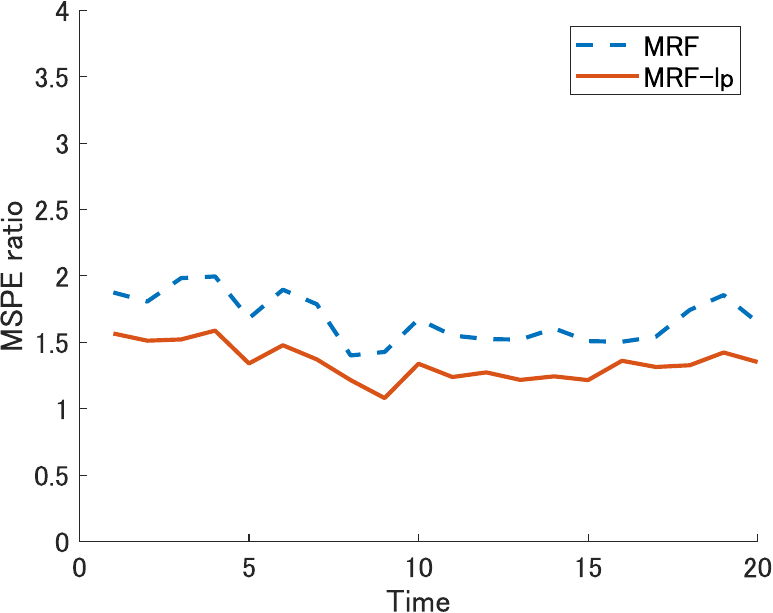}}
    \end{minipage}
    \hspace{1cm}
    \begin{minipage}[t]{6cm}
   \subfigure[Baseline\\ (MRF: 1.466, MRF-lp: 1.269)]{
  \includegraphics[width = 5.2cm,pagebox=artbox,clip]{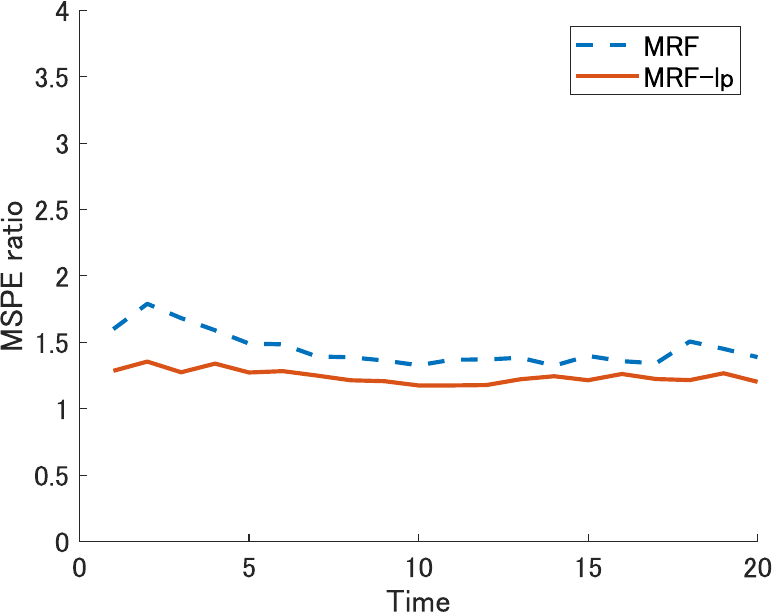}}\\
 \subfigure[Small sample case\\ (MRF: 1.225, MRF-lp: 1.114)]{
  \includegraphics[width = 5.2cm,pagebox=artbox,clip]{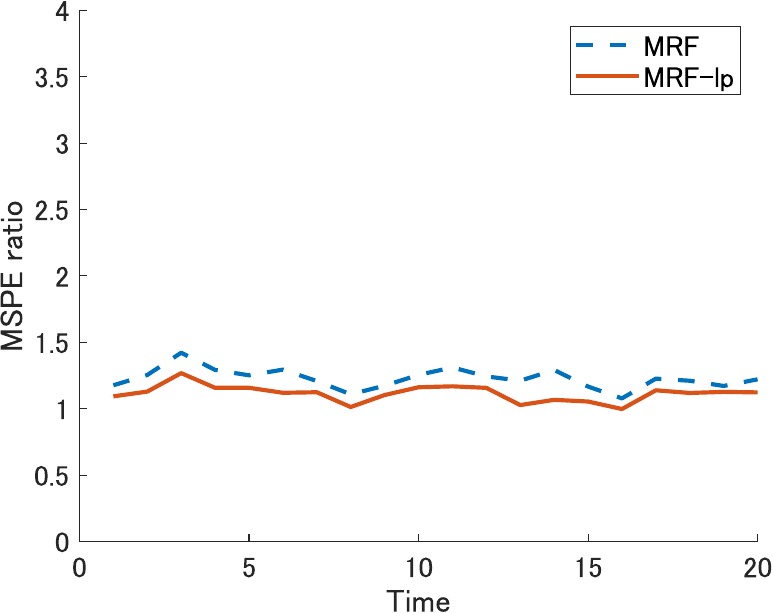}}\\
 \subfigure[Low noise case\\ (MRF: 1.625, MRF-lp: 1.372)]{
  \includegraphics[width = 5.2cm,pagebox=artbox,clip]{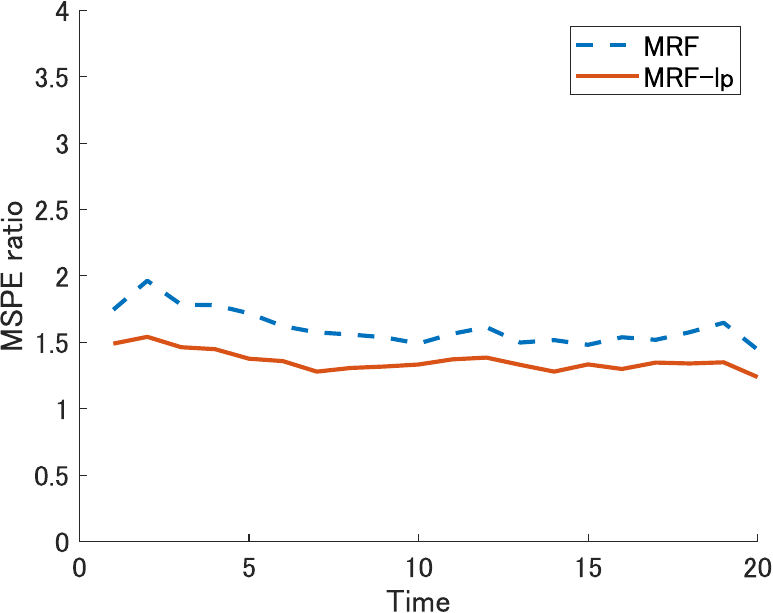}}\\
  \subfigure[Smooth case\\ (MRF: 1.178, MRF-lp: 1.125)]{
  \includegraphics[width = 5.2cm,pagebox=artbox,clip]{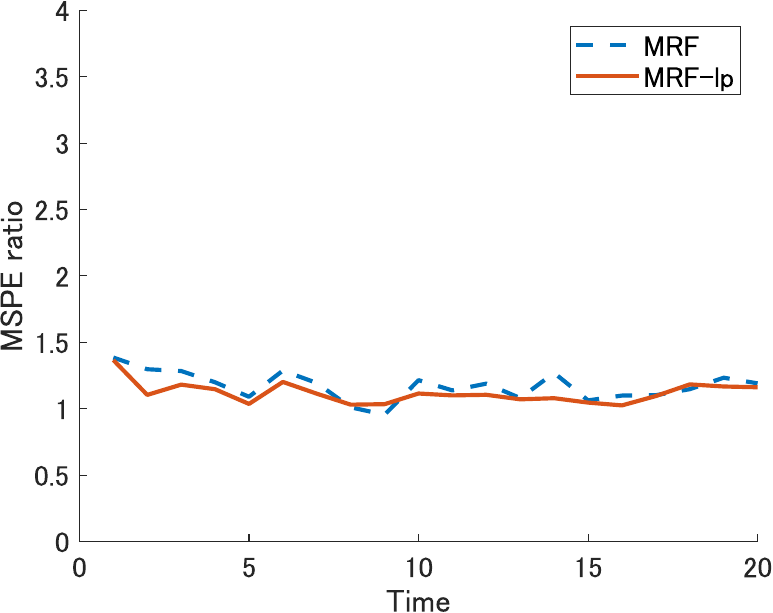}}
    \end{minipage}
  \caption{
  Comparison of the ratio of the MSPE in the linear and Gaussian case. (a)--(d) $M=2$ and (e)--(h) $M=4$. The ratio of the totally averaged MSPE at all of the spatio-temporal grid points is given within parentheses. The relative times of the MRF and MRF-lp were 0.308 and 0.516 in $M=2$ and 0.619 and 0.806 in $M=4$, respectively.  
  }
\label{fig_linear_gaussian_comparison} 
\end{figure}

The results are summarized in Figure \ref{fig_linear_gaussian_comparison}. 
As for the ratio of the MSPE at each time point, the difference between the MRF and the MRF-lp was smaller in $M=4$ than in $M=2$ because the forecast covariance matrix was sufficiently approximated in the MRF and MRF-lp with $M=4$. When $M=4$, the MRF showed a smaller averaged MSPE than the MRF-lp at a few time points, but the totally averaged MSPE of the MRF-lp was smaller than that of the MRF in all cases. The computational time of the MRF-lp was larger than that of the MRF but smaller than that of the original Kalman filter. Consequently, our proposed MRF-lp worked well.

Second, Algorithm \ref{alg_mrf_lp_linear_non-gaussian} was evaluated for linear and non-Gaussian state-space models. We employed the Gamma and Poisson distributions as $g_{ti}$. For the Gamma distribution,
\begin{align*}
  g_{ti}(y_{ti} \mid x_{ti}) &= \frac{b_{ti}^a}{\Gamma (a)} y_{ti}^{a-1} \exp(-b_{ti} y_{ti}) \quad (b_{ti} = a \exp(-x_{ti})),\\
  \left( \bm{u}_{\bm{x}_t} \right)_i &= a \{-1+y_{ti} \exp(-x_{ti})\}, \quad i \in I_{\bm{y}_t},
  \\
  \left( D_{\bm{x}_t} \right)_{ij} &= \frac{a y_{ti}}{\exp(x_{ti})}, \quad i,j \in I_{\bm{y}_t} \; \mbox{and} \; i=j.
\end{align*}
We took $a=3$ as the shape parameter. For the Poisson distribution,
\begin{align*}
  g_{ti}(y_{ti} \mid x_{ti}) &= \frac{\theta_{ti}^{y_{ti}} \exp(-\theta_{ti})}{y_{ti}!} \quad (\theta_{ti} = \exp(x_{ti})),\\
  \left( \bm{u}_{\bm{x}_t} \right)_i &= y_{ti} - \exp(x_{ti}), \quad i \in I_{\bm{y}_t},
  \\
  \left( D_{\bm{x}_t} \right)_{ij} &= \exp(x_{ti}), \quad i,j \in I_{\bm{y}_t} \; \mbox{and} \; i=j.
\end{align*}
In both cases, the elements of $D_{\bm{x}_t}$ are nonnegative. The count of iterations for calculating the evaluation measures was 30. $A_t$ was the same as in the first simulation.

\begin{figure}[t]
  \centering
  \begin{minipage}[t]{6cm}
 \subfigure[Gamma case\\ (MRF: 1.627, MRF-lp: 1.295)]{
  \includegraphics[width = 5.2cm,pagebox=artbox,clip]{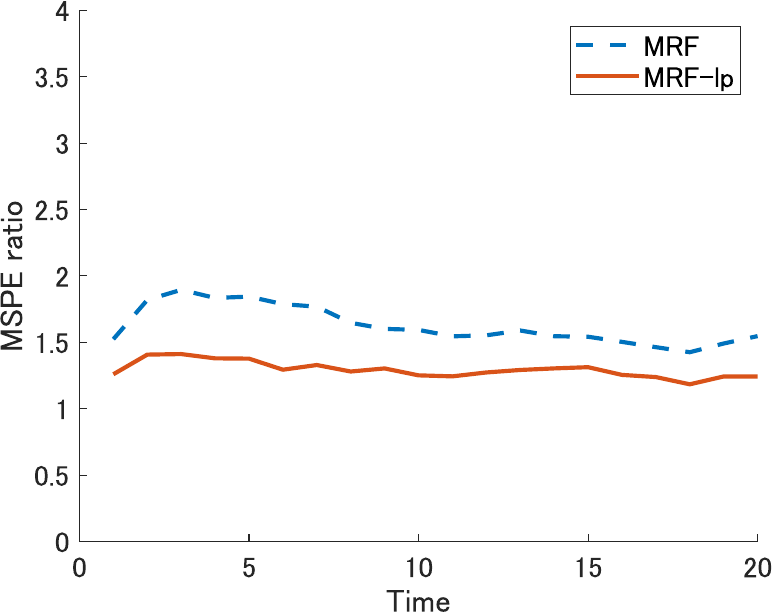}}\\
 \subfigure[Poisson case\\ (MRF: 1.388, MRF-lp: 1.160)]{ 
  \includegraphics[width = 5.2cm,pagebox=artbox,clip]{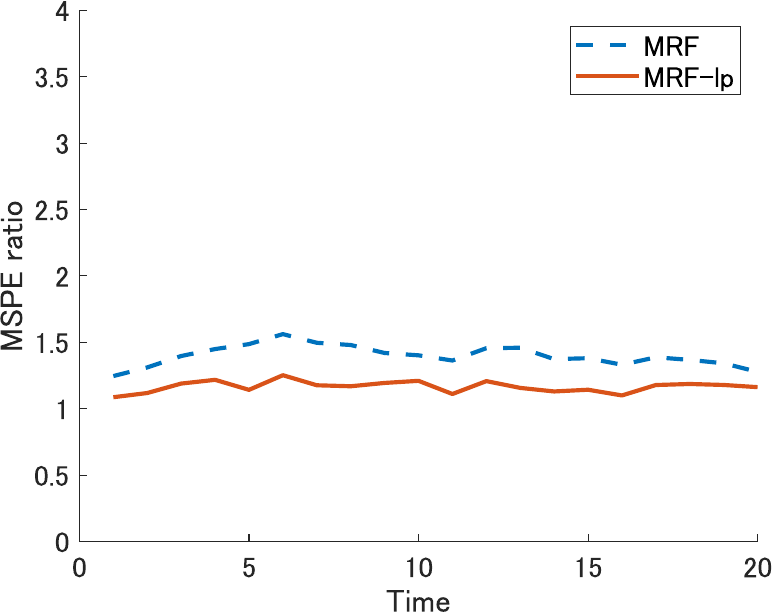}}
    \end{minipage}
    \hspace{1cm}
    \begin{minipage}[t]{6cm}
   \subfigure[Gamma case\\ (MRF: 1.181, MRF-lp: 1.121)]{
  \includegraphics[width = 5.2cm,pagebox=artbox,clip]{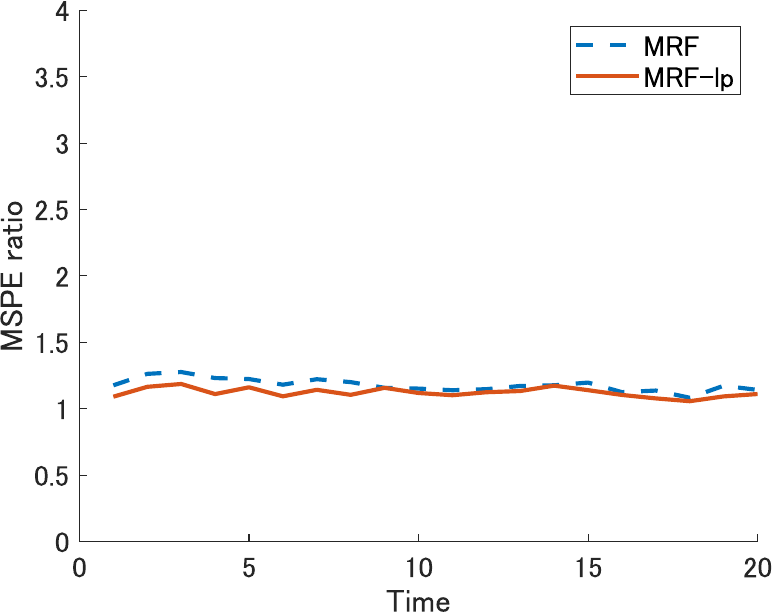}}\\
 \subfigure[Poisson case\\ (MRF: 1.153, MRF-lp: 1.086)]{
  \includegraphics[width = 5.2cm,pagebox=artbox,clip]{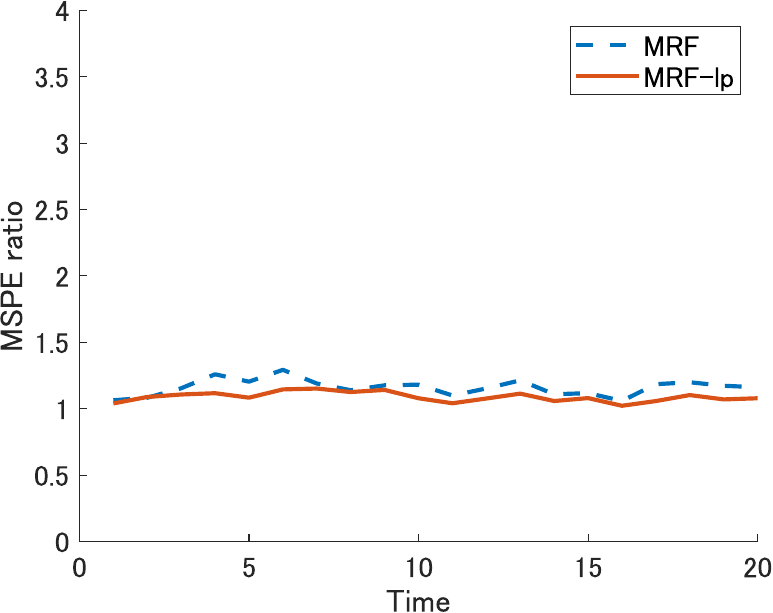}}
    \end{minipage}
  \caption{Comparison of the ratio of the MSPE in the linear and non-Gaussian case. (a) and (b)  $M=2$ and (c) and (d) $M=4$. The ratio of the totally averaged MSPE at all of the spatio-temporal grid points is given within parentheses. The relative times of the MRF and MRF-lp in the Poisson case were 0.339 and 0.394 in $M=2$ and 0.569 and 0.658 in $M=4$, respectively. 
  }
\label{fig_linear_non-gaussian_comparison} 
\end{figure}

The results of Figure \ref{fig_linear_non-gaussian_comparison} were similar to those of the first simulation except for the computational time. The difference in the computational time between the MRF and the MRF-lp was smaller than that of the first simulation. This is because the original method includes the computation of the large inverse matrix at each step of the Newton--Raphson update \eqref{eq_newton-raphson_update}, and
the MRF and the MRF-lp have a greater impact on reducing the computational time of the original method than the first simulation. Consequently, the computational time of the original method was relatively large in comparison with those of the MRF and MRF-lp, leading to a small difference in the computational time between the MRF and the MRF-lp.

Third, Algorithm \ref{alg_mrf_lp_nonlinear_non-gaussian} was evaluated for nonlinear and non-Gaussian state-space models. To derive $\mathcal{A}_t$, building upon \cite{Jurek_2022}, we used the system of partial differential equations in \cite{Lorenz_2005} which is defined on a regular grid on a unit circle. Therefore, in the third simulation, $\bm{x}_t$ was on the $n_\mathcal{G}=1156$ grid points on the unit circle $D_0$. By using the fourth-order Runge--Kutta method, we obtained $\mathcal{A}_t$ and $A_t^{\prime}$. In this case, $A_t^{\prime}$ was a sparse matrix. The time interval in the fourth-order Runge--Kutta method was 0.5, so that $T=20$ means the total number of time steps. A description of deriving $\mathcal{A}_t$ and $A_t^{\prime}$ is given in Section \ref{suppl_lorenz_2005}. The count of iterations and $g_{ti}$ were the same as in the second simulation.

\begin{figure}[t]
  \centering
  \begin{minipage}[t]{6cm}
 \subfigure[Gamma case\\ (MRF: 1.228, MRF-lp: 1.075)]{
  \includegraphics[width = 5.2cm,pagebox=artbox,clip]{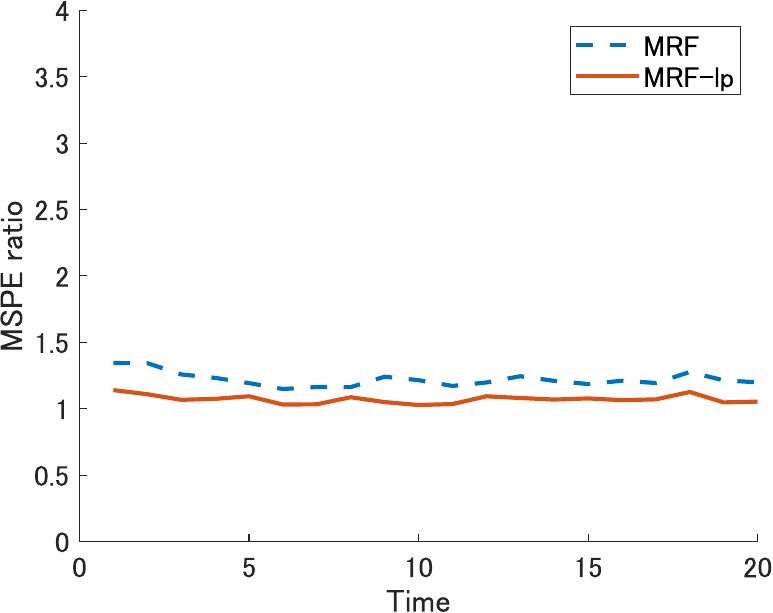}}\\
 \subfigure[Poisson case\\ (MRF: 1.183, MRF-lp: 1.044)]{ 
  \includegraphics[width = 5.2cm,pagebox=artbox,clip]{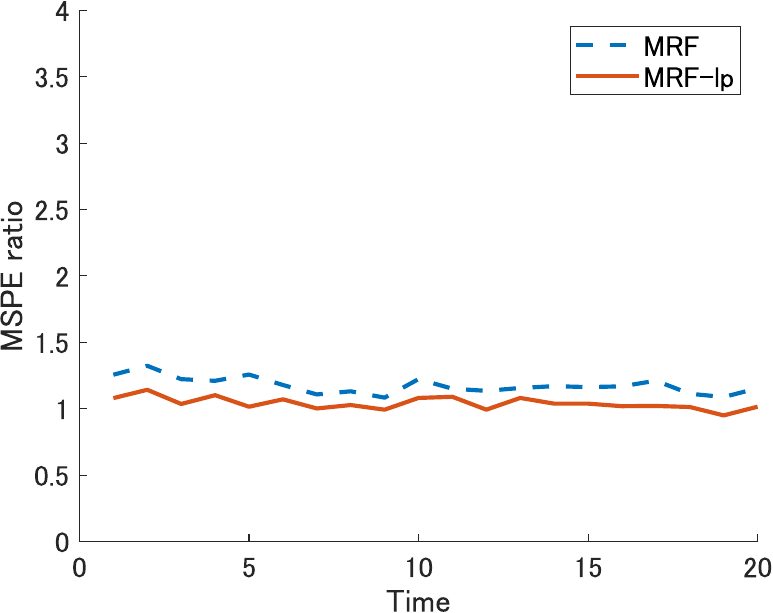}}
    \end{minipage}
    \hspace{1cm}
    \begin{minipage}[t]{6cm}
   \subfigure[Gamma case\\ (MRF: 1.050, MRF-lp: 1.035)]{
  \includegraphics[width = 5.2cm,pagebox=artbox,clip]{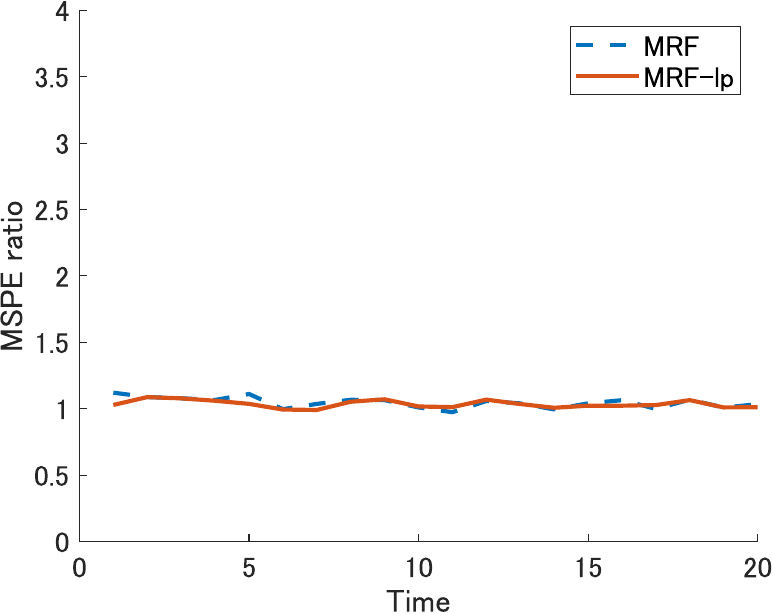}}\\
 \subfigure[Poisson case\\ (MRF: 1.044, MRF-lp: 1.021)]{
  \includegraphics[width = 5.2cm,pagebox=artbox,clip]{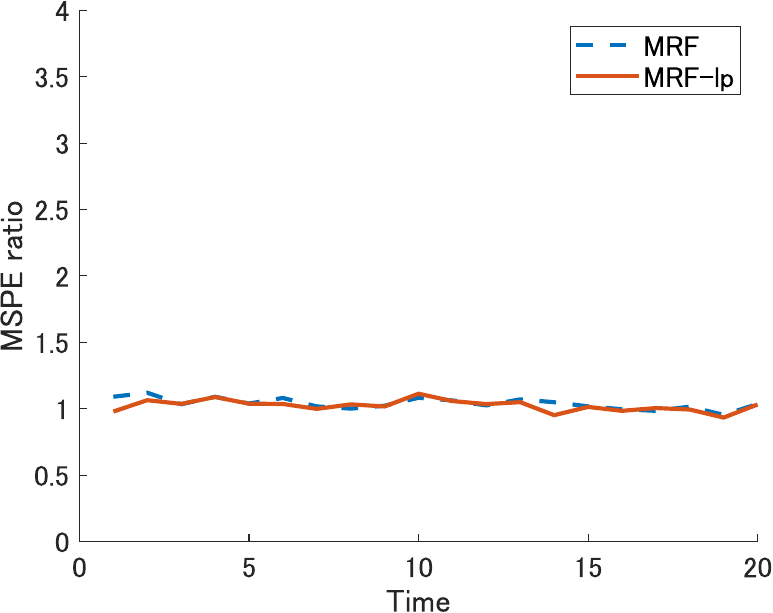}}
    \end{minipage}
  \caption{Comparison of the ratio of the MSPE in the nonlinear and non-Gaussian case. (a) and (b) $M=2$ and (c) and (d) $M=4$. The ratio of the totally averaged MSPE at all of the spatio-temporal grid points is given within parentheses. The relative times of the MRF and MRF-lp in the Poisson case were 0.811 and 0.818 in $M=2$ and 0.847 and 0.872 in $M=4$, respectively. 
  }
\label{fig_nonlinear_non-gaussian_comparison} 
\end{figure}

Figure \ref{fig_nonlinear_non-gaussian_comparison} displays the better performance of the MRF-lp than that of the MRF similar to the second simulations in light of the MSPE. Particularly, the MRF-lp had almost the same accuracy as the original method for $M=2,4$. However, for the MRF and MRF-lp, the reduction in the computational time relative to the original method was less than 20 \%. This is because the calculation of $\mathcal{A}_t$ and $A_t^{\prime}$ resulted in the computational burden in the forecast step of each method, which diminished the degree of the reduction of the computational time 
in the update step of the MRF and MRF-lp.

Fourth, we compared the MRF and the MRF-lp with the EnKF. In accordance with the setting of \cite{Jurek_2022}, we adopted the EnKF for nonlinear and Gaussian state-space models 
with $\mathcal{A}_t$, $n_\mathcal{G}$, and $D_0$, which were the same as those of the third simulation. 
The EnKF entails the number of ensemble members $M^{\prime}$ and an estimate of the forecast covariance matrix $C_t$. As the canonical setting, we selected $C_t = \tilde{S}_t \circ T_t$ where $\tilde{S}_t$ is the sample covariance matrix of the ensemble members, $T_t$ is a sparse correlation matrix, and the symbol ``$\circ$'' refers to the Hadamard product. $T_t$ was generated from Kanter's function \citep{Kanter_1997}. This selection of $C_t$ is called the localization, also known as the covariance tapering. To ensure a fair comparison, we set $M^{\prime} = N = \sum_{m=0}^{M} r_{j_1,\ldots,j_m}^{\prime}$, 
and the tapering radius in $T_t$ was adjusted such that $C_t$ has $N$ nonzero entries per row. Consequently, we considered $N=30, 40$ because of the baseline for the MRF and MRF-lp. In addition, the covariance functions in the baseline and the smooth case were used. The algorithm of the EnKF is described in \ref{append_comparison_ensemble_kf}. 

\begin{figure}[t]
  \centering
  \begin{minipage}[t]{6cm}
 \subfigure[Baseline\\ (MRF: 1.809, MRF-lp: 1.437,\\ EnKF: 1.376)]{
  \includegraphics[width = 5.2cm,pagebox=artbox,clip]{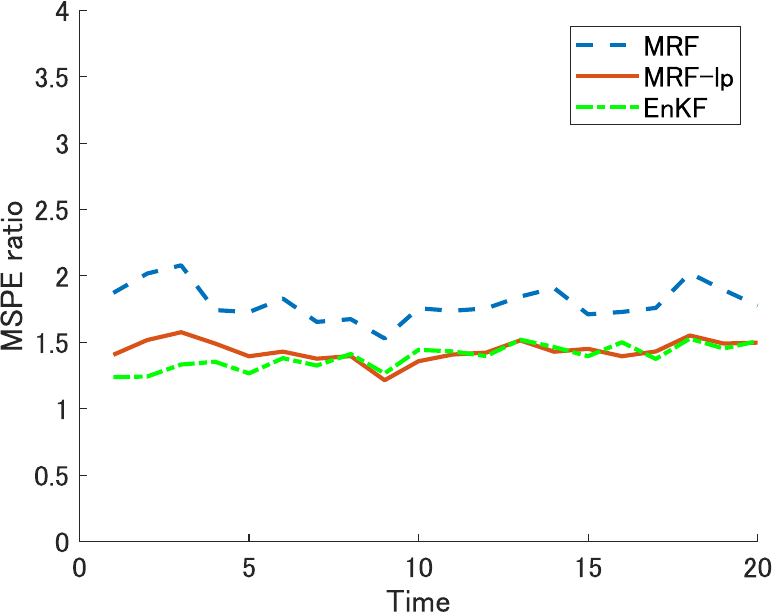}}\\
 \subfigure[Smooth case\\ (MRF: 1.325, MRF-lp: 1.018,\\ EnKF: 3.590)]{ 
  \includegraphics[width = 5.2cm,pagebox=artbox,clip]{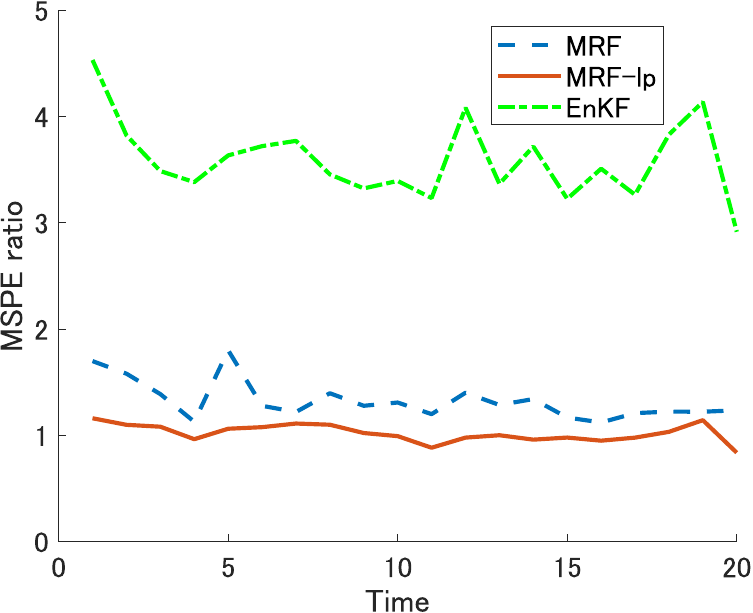}}
    \end{minipage}
    \hspace{1cm}
    \begin{minipage}[t]{6cm}
   \subfigure[Baseline\\ (MRF: 1.236, MRF-lp: 1.154,\\ EnKF: 1.242)]{
  \includegraphics[width = 5.2cm,pagebox=artbox,clip]{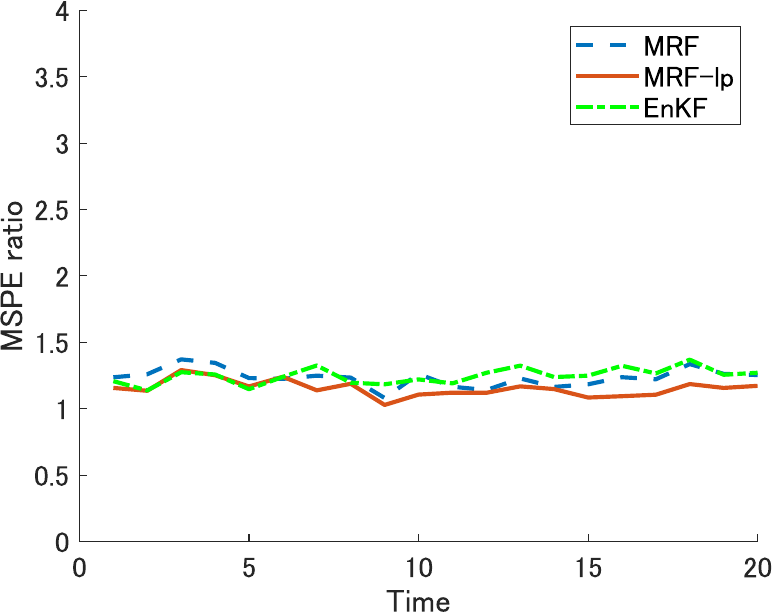}}\\
 \subfigure[Smooth case\\ (MRF: 1.049, MRF-lp: 1.012,\\ EnKF: 2.856)]{
  \includegraphics[width = 5.2cm,pagebox=artbox,clip]{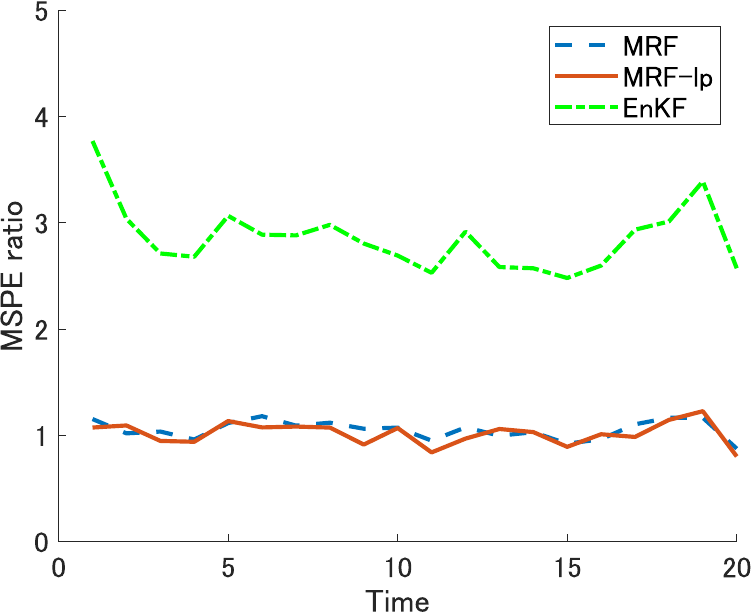}}
    \end{minipage}
  \caption{Comparison of the ratio of the MSPE in the nonlinear and Gaussian case. (a) and (b)  $M=2$, $N=30$ and (c) and (d) $M=4$, $N=40$. The ratio of the totally averaged MSPE at all of the spatio-temporal grid points is given within parentheses. The relative times of the MRF, MRF-lp, and EnKF in the baseline were 0.898, 0.920, and 0.709 in $M=2$, $N=30$ and 0.936, 0.959, and 0.906 in $M=4$, $N=40$, respectively. 
  }
\label{fig_nonlinear_gaussian_comparison_enkf} 
\end{figure}

As shown in Figure \ref{fig_nonlinear_gaussian_comparison_enkf}, the EnKF attained the desirable approximation accuracy in the shorter computation time compared with the MRF and MRF-lp in the baseline. However, in the smooth case, the approximation accuracy of the EnKF degenerated remarkably because the localization ignored the strong spatial correlation. The MRF and MRF-lp worked well in both cases, and the approximation accuracy of the MRF-lp was better than that of the MRF by the somewhat additional computation time. Moreover, each iteration of the EnKF requires the implementation of computation of $\mathcal{A}_t$ as many times as the number of ensemble members, but we need to compute $\mathcal{A}_t$ and $A_t^{\prime}$ 
only once in each iteration of the MRF and MRF-lp. Hence, if large $M^{\prime}$ is required and the computational time of $\mathcal{A}_t$ is large, the EnKF can give rise to the computational burden, unlike the MRF and the MRF-lp.

Fifth, we examined the scalability of the MRF and MRF-lp for large spatio-temporal datasets in the linear and Gaussian case and the nonlinear and non-Gaussian case. Poisson distribution in the third simulation was employed as the non-Gaussian case. $n_\mathcal{G}$ was selected from 900 to 10,000, and we set $T=50$ and used the baseline of the MRF and MRF-lp in $M=4$. Furthermore, unlike the previous simulations, we considered $A_t = 0.6 \bm{\mathrm{I}}_{n_\mathcal{G}}$ and $(\mathcal{A}_t(\bm{x}))_{i} = 0.1 x_{i}+0.1 x_{i}^2$ to focus as much as possible on the effect of the fast computation in the update step of Algorithms \ref{alg_mrf_lp} and \ref{alg_mrf_lp_nonlinear_non-gaussian}. The relative computational time and the ratio of the totally averaged MSPE at all of the spatio-temporal grid points relative to that of the original method were calculated from one iteration. 

\begin{table}[t]
\caption{Comparison of the computational time.}
\label{tab_comparison_time_mspe}
\vspace{-10pt}
\begin{center}
\begin{tabular}{llcccccc} \toprule
 &  & \multicolumn{2}{c}{$n_\mathcal{G}=900$} & \multicolumn{2}{c}{$n_\mathcal{G}=1,764$} & \multicolumn{2}{c}{$n_\mathcal{G}=2,704$} \\ \midrule
\multicolumn{1}{l}{Linear} & MRF & \multicolumn{2}{c}{0.860 (1.214)} & \multicolumn{2}{c}{0.373 (1.325)} & \multicolumn{2}{c}{0.201 (1.404)} \\
\multicolumn{1}{l}{and Gaussian} & MRF-lp & \multicolumn{2}{c}{1.049 (1.167)} & \multicolumn{2}{c}{0.471 (1.243)} & \multicolumn{2}{c}{0.236 (1.321)} \\
\multicolumn{1}{l}{Nonlinear} & MRF & \multicolumn{2}{c}{0.661 (1.005)} & \multicolumn{2}{c}{0.375 (1.016)} & \multicolumn{2}{c}{0.258 (1.027)} \\
\multicolumn{1}{l}{and non-Gaussian} & MRF-lp & \multicolumn{2}{c}{0.803 (1.004)} & \multicolumn{2}{c}{0.400 (1.011)} & \multicolumn{2}{c}{0.271 (1.010)} \\ \cmidrule[1pt]{1-8}
 &  & \multicolumn{2}{c}{$n_\mathcal{G}=3,600$} & \multicolumn{2}{c}{$n_\mathcal{G}=5,329$} & \multicolumn{2}{c}{$n_\mathcal{G}=10,000$} \\ \midrule
\multicolumn{1}{l}{Linear} & MRF & \multicolumn{2}{c}{0.169 (1.484)} & \multicolumn{2}{c}{0.115 (1.621)} & \multicolumn{2}{c}{0.060 (1.857)} \\
\multicolumn{1}{l}{and Gaussian} & MRF-lp & \multicolumn{2}{c}{0.187 (1.372)} & \multicolumn{2}{c}{0.124 (1.516)} & \multicolumn{2}{c}{0.063 (1.710)} \\
\multicolumn{1}{l}{Nonlinear} & MRF & \multicolumn{2}{c}{0.198 (1.026)} & \multicolumn{2}{c}{0.148 (1.030)} & \multicolumn{2}{c}{0.080 (1.078)} \\
\multicolumn{1}{l}{and non-Gaussian} & MRF-lp & \multicolumn{2}{c}{0.215 (1.017)} & \multicolumn{2}{c}{0.153 (1.011)} & \multicolumn{2}{c}{0.081 (1.031)} \\ \bottomrule
\end{tabular}
\end{center}

\vspace{-10pt}
The ratio of the totally averaged MSPE at all of the spatio-temporal grid points relative \newline to that of the original method is given within parentheses.
\end{table}

The results, reported in Table \ref{tab_comparison_time_mspe}, indicate better scalability of the MRF and MRF-lp than that of the original method. Especially in large $n_\mathcal{G}$, the computational time of the MRF-lp was somewhat larger than that of the MRF, but the MRF-lp improved the MSPE when compared with the MRF. As $n_\mathcal{G}$ was larger, the MSPE was larger, and the computational time was smaller. 
This is because the settings of the MRF and MRF-lp are fixed for all cases of $n_\mathcal{G}$. 
In the nonlinear and non-Gaussian case, the MRF and MRF-lp showed the MSPE ratio close to one in the very short computation time for these fixed settings.

Sixth, we checked the numerical stability of the MRF and MRF-lp and specified a particular scenario where the MRF-lp outperforms the MRF in terms of the MSPE. In this simulation study, we focused on the linear and Gaussian case with the Mat\'ern covariance function with the smoothness parameter 3.5 and the Gaussian (squared exponential) covariance function. The range
parameters were 1, 3, 5, and 7, and the count of iterations for calculating the ratios of the totally averaged MSPE at all of the spatiotemporal grid points was 30. When a pair of observations has a negligible ($<$ 0.05) correlation, the distance between the observations is called the effective range. For the range parameters 1, 3, 5, and 7, their effective ranges were 6.88, 2.30, 1.38, and 0.98, respectively, in the Mat\'ern covariance function with the smoothness parameter 3.5. Similarly, in the Gaussian covariance function, they were 1.74, 1.0, 0.78, and 0.66 for the range parameters 1, 3, 5, and 7, respectively. For the MRF, we adopted $(r_0,r_{j_1},r_{j_1,j_2},r_{j_1,j_2,j_3},r_{j_1,j_2,j_3,j_4})=(r_0^{\prime},r_{j_1}^{\prime},r_{j_1,j_2}^{\prime},r_{j_1,j_2,j_3}^{\prime},r_{j_1,j_2,j_3,j_4}^{\prime}) = (10,10,10,2,2)$. For the MRF-lp, we set the same knot case and the same rank case, that is, $(r_0,r_{j_1},r_{j_1,j_2},r_{j_1,j_2,j_3},r_{j_1,j_2,j_3,j_4}) = (10,10,10,2,2)$ and $(r_0^{\prime},r_{j_1}^{\prime},r_{j_1,j_2}^{\prime},r_{j_1,j_2,j_3}^{\prime},r_{j_1,j_2,j_3,j_4}^{\prime}) = (7,7,7,2,2)$ (Case 1) and $(r_0,r_{j_1},r_{j_1,j_2},r_{j_1,j_2,j_3},r_{j_1,j_2,j_3,j_4}) = (20,20,20,4,4)$ and $(r_0^{\prime},r_{j_1}^{\prime},r_{j_1,j_2}^{\prime},r_{j_1,j_2,j_3}^{\prime},r_{j_1,j_2,j_3,j_4}^{\prime}) = (10,10,10,2,2)$ (Case 2). The other conditions were the same as the baseline of the linear and Gaussian case 
in Section \ref{subsec_simulation_study}.

\begin{table}[t]
\caption{Comparative performance of the MRF and MRF-lp with respect to the MSPE under the strong spatial correlation.}
  \label{tab_numerical_stability}
\begin{center}
\begin{tabular}{llccc} \toprule
Covariance & Range & MRF & MRF-lp (Case 1) & MRF-lp (Case 2) \\ \midrule
Mat\'{e}rn (3.5) & 1 & - & 1.034 & 1.019 \\
 & 3 & 1.041 & 1.084 & 1.022 \\
 & 5 & 1.095 & 1.139 & 1.094 \\
 & 7 & 1.121 & 1.199 & 1.116 \\
Gaussian & 1 & - & 1.042 & 1.028 \\
 & 3 & 1.107 & 1.138 & 1.067 \\
 & 5 & 1.125 & 1.204 & 1.110 \\
 & 7 & 1.172 & 1.246 & 1.138 \\ \bottomrule
\end{tabular}
\end{center}
\end{table}

Table \ref{tab_numerical_stability} shows the MSPE in each case. 
In the same covariance case of Table \ref{tab_numerical_stability}, the larger the effective range was, the smaller the MSPE was. The large effective range means that the large-scale spatial variation is dominant. Therefore, this may be because the MRA and the MRA-lp, respectively, extend the predictive process \citep[][]{Banerjee_2008} and the linear projection \citep[][]{Banerjee_2013} which are originally effective for fitting the large-scale spatial variation. 
Moreover, we cannot conduct the MRF in the cases of the range parameter 1 because the output of the MRA was unstable and included outliers due to the numerical instability of the calculation related to the inverse matrix of $\widehat{V}_{j_1,\ldots,j_m}^{m}$ ($m=0,\ldots,M$) in the MRA implemented in the MRF, that is, Algorithm 2 in Algorithm 3 with $\Phi_{j_1,\ldots,j_m} = \bm{\mathrm{I}}_{r_{j_1,\ldots, j_m}}$. This numerical stability is evaluated by the condition number which means the ratio of the largest $\sigma_l$ and the smallest $\sigma_s$ eigenvalues of the positive definite matrix \citep[see][]{Dixon_1983} and evaluates how well the positive definite matrix is conditioned. If a positive definite matrix is ill-conditioned, the calculation of the inverse matrix may be unstable with the propagation of round-off errors due to the finite precision arithmetic. The condition number closer to 1 indicates better numerical stability. We compared the condition number of $\widehat{V}_{j_1,\ldots,j_m}^{m}$ ($m=0,\ldots,M$) between the MRF and the MRF-lp. This simulation study is similar to those in Section 3.2 of \cite{Banerjee_2013} and Section 4.3 of \cite{Hirano_2021}. 
We confirmed that as the number $r_{j_1,\ldots,j_m}$ of knots, $m$, the rank $r_{j_1,\ldots,j_m}^{\prime}$ of $\Phi_{j_1,\ldots, j_m}$, the smoothness parameter, and the effective range increased, the condition numbers of $\widehat{V}_{j_1,\ldots,j_m}^{m}$ tended to increase. However, our proposed MRF-lp empirically showed smaller conditional numbers than those of the MRF. This may be because the MRA-lp introduces $\Phi_{j_1,\ldots,j_m}$, the minor eigenvalues of $V_{j_1,\ldots,j_m}^{m}$ can be removed in the selection of $\Phi_{j_1,\ldots,j_m}$ (see Section \ref{subsec_select_phi} for details), and Proposition \ref{prop_mra-lp_inv} holds. These results are not reported here. In the cases of the range parameter 1, their effective ranges were large relative to the size of
the domain $[0,1]^2$. However, in the case of the exponential covariance function, we can conduct the MRF for the appropriate knots even when the effective range is greater than 100. This suggests that both the large smoothness parameter and the large effective range degrade the numerical stability of the MRF. The MRF-lp can select more knots and improve the MSPE under the strong spatial correlation with the large smoothness parameter and effective range. On the other hand, unlike the MRF-lp, it is difficult for the MRF to improve the MSPE to some extent by adding more knots because increasing knots often causes an infeasible result due to the high condition number of $\widehat{V}_{j_1,\ldots,j_m}^{m}$. Actually, in the case of the range parameter 1, we needed to reduce the number of knots of the MRF in order to implement the MRF normally. In this case, most of the feasible trials of the MRF produced larger MSPEs than those of the MRF-lp (Case 1), and all of their MSPEs by the MRF 
were larger than a feasible MSPE by the MRF-lp. It is suggested that it is difficult to set the appropriate knot size in the MRF under the strong spatial correlation. We also obtained similar results for the nonlinear and non-Gaussian case. 

\subsection{Real data analysis\label{subsec_real_data}}

In this subsection, we discuss the results when we applied our proposed MRF-lp to hourly measurements of total precipitable water (TPW)  previously analyzed in \cite{Katzfuss_2017b} and 
\cite{Jurek_2022, Jurek_2023}. 
This dataset means the amount of the column water vapor on an area and is made by the Microwave Integrated Retrieval System (MiRS), which was developed and supported by the National Oceanic and Atmospheric Administration (NOAA) \citep[e.g.,][]{Lee_2022}. Water vapor is one of the important elements in the earth's water cycle, and the analysis of TPW data is helpful for numerical weather prediction. We selected the TPW data inside the rectangular region $[-93.6937, -79.4621] \times [32.6534, 43.7983]$ at $T=8$ points in January 2011, leading to $100 \times 100$ grid points. This time interval differs from that used in \cite{Jurek_2022, Jurek_2023}. 

Building upon \cite{Jurek_2023}, 
we determined the settings of our real data analysis. We assumed \eqref{eq_observation_equation} and \eqref{eq_state_equation} for the TPW data. As a result, $\bm{\mu}_{t|t}$ was calculated efficiently by using Algorithm \ref{alg_mrf_lp}. In deriving $A_t$, we employed the advection-diffusion equation (see Section \ref{suppl_advection_diffusion} for details). $Q_t$ and $\Sigma_{0|0}$ were created by the Mat\'{e}rn covariance function with the smoothness parameter 1.5 and the range parameter $\lambda$, and their variances were $\sigma_{\bm{w}}^2$ and $\sigma_0^2$, respectively. $R_t$ was the diagonal matrix with the variance $\sigma_{\bm{v}}^2$. 
To determine $\lambda$, $\sigma_{\bm{w}}^2$, $\sigma_{\bm{v}}^2$, and $\sigma_0^2$, we subtracted the sample mean from $\bm{y}_t$ for each $t$, conducted the maximum likelihood estimation at each time point by treating its adjusted $\bm{y}_t$ as the sum of purely Gaussian spatial data with the Mat\'{e}rn covariance function and the measurement error, and adopted the averaged values of the range and the variance of the measurement error over time for $\lambda$ and $\sigma_{\bm{v}}^2$. Moreover, by assuming that the averaged value of the variance of purely Gaussian spatial data over time is equal to $\sigma_0^2 + \sigma_{\bm{w}}^2$ satisfying $\sigma_0^2 = 9 \sigma_{\bm{w}}^2$, $\sigma_0^2$ and $\sigma_{\bm{w}}^2$ were determined. For $\bm{x}_0$, we took $\bm{\mu}_{0|0}$, each element of which is the averaged value of elements of $\bm{y}_1$.

For the MRF-lp, we set $M=4$, $J_m = 4$, $(r_{0}, r_{j_1},r_{j_1,j_2},r_{j_1, j_2, j_3},r_{j_1, j_2, j_3,j_4}) = (100,50,50,20,20)$, and $(r^{\prime}_{0}, r^{\prime}_{j_1},r^{\prime}_{j_1,j_2},r^{\prime}_{j_1, j_2, j_3},r^{\prime}_{j_1, j_2, j_3,j_4}) = (30,20,10,10,7)$. The MRF had $M=4$, $J_m = 4$, \\
$(r_{0}, r_{j_1},r_{j_1,j_2},r_{j_1, j_2, j_3},r_{j_1, j_2, j_3,j_4}) = (r^{\prime}_{0}, r^{\prime}_{j_1},r^{\prime}_{j_1,j_2},r^{\prime}_{j_1, j_2, j_3},r^{\prime}_{j_1, j_2, j_3,j_4}) = (30,20,10,10,7)$. We compared the MRF and MRF-lp with the Kalman filter by assessing the mean squared sum of the difference between the filtering means of the Kalman filter and approximation method and the ratio of the averaged computational time of the forecast and update steps relative to that of the Kalman filter. Additionally, we displayed $\bm{\mu}_{t|t}$ for three methods at $t=1,4,8$.

\begin{figure}[t]
  \centering
    \includegraphics[width = 12.2cm,pagebox=artbox,clip]{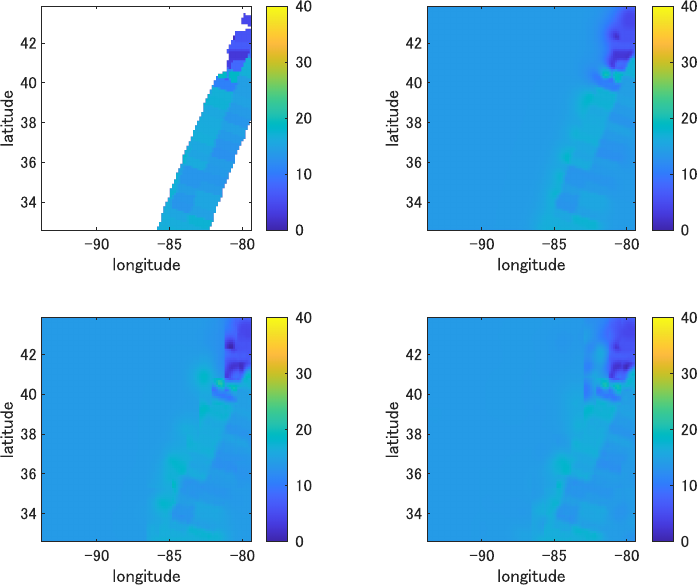}
    \caption{
    The filtering means at $t=1$. The real data is on the upper left. We adopt the Kalman filter 
    for the right column in the top row. The MRF and MRF-lp are used for the left and right columns in the bottom row, respectively.}
  \label{fig_prediction_surface_t1} 
\end{figure}
  
  \begin{figure}[t]
    \centering
      \includegraphics[width = 12.2cm,pagebox=artbox,clip]{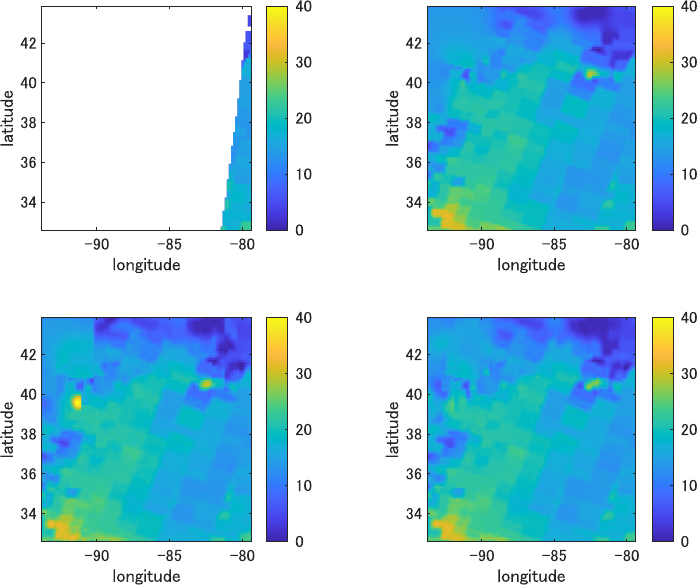}
      \caption{
        The filtering means at $t=4$. The real data is on the upper left. We adopt the Kalman filter 
      for the right column in the top row. The MRF and MRF-lp are used for the left and right columns in the bottom row, respectively.}
    \label{fig_prediction_surface_t4} 
\end{figure}
  
\begin{figure}[t]
      \centering
        \includegraphics[width = 12.2cm,pagebox=artbox,clip]{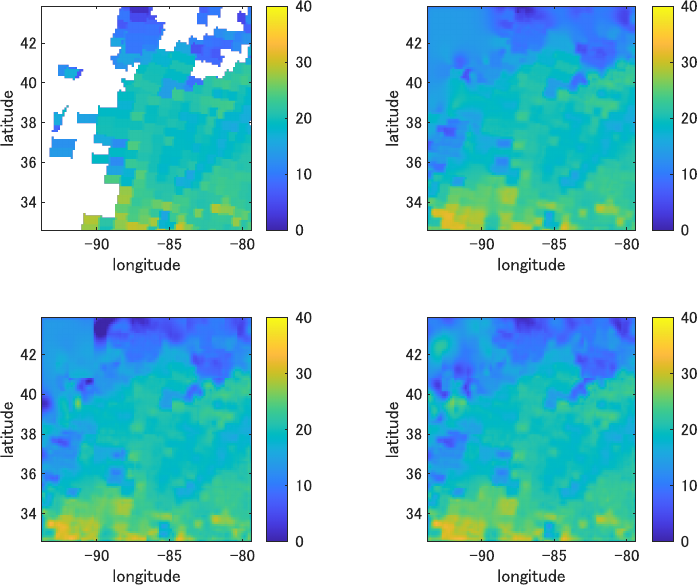}
        \caption{
          The filtering means at $t=8$. The real data is on the upper left. We adopt the Kalman filter 
        for the right column in the top row. The MRF and MRF-lp are used for the left and right columns in the bottom row, respectively.}
      \label{fig_prediction_surface_t8} 
\end{figure}

Figures \ref{fig_prediction_surface_t1}--\ref{fig_prediction_surface_t8} show the results of our real data analysis. For the MRF and MRF-lp, the pairs of the mean squared sum and relative computational time were $(1.6065, 0.4297)$ and $(1.0003, 0.4413)$, respectively. This result suggests that our proposed MRF-lp has competitive efficiency. As shown in Figure \ref{fig_prediction_surface_t4}, unlike the MRF-lp, the filtering mean of the MRF described the local discrepancy from that of the Kalman filter in the middle of the time interval, 
 which decreased the approximation accuracy of the MRF in comparison with the MRF-lp. However, in Figure \ref{fig_prediction_surface_t8}, this local discrepancy in the MRF eventually disappeared. To improve the local fitting, the MRF needs to increase $r_{j_1, j_2, j_3,j_4} = r^{\prime}_{j_1, j_2, j_3,j_4}$, but this can sharply increase the computational burden. On the other hand, since our proposed MRF-lp can control $r_{j_1, j_2, j_3,j_4}$ and $r^{\prime}_{j_1, j_2, j_3,j_4}$ separately, we can efficiently improve the local fitting by increasing only $r_{j_1, j_2, j_3,j_4}$. Furthermore, although there are some artifacts for both the MRF and the MRF-lp in Figures \ref{fig_prediction_surface_t1}--\ref{fig_prediction_surface_t8} 
due to the partition of the spatial domain in Algorithm \ref{alg_approx_cov_mat}, it seems that the MRF-lp mitigates the degree of the artifact in this real data analysis. A taper version of Algorithm \ref{alg_approx_cov_mat} based on \cite{Katzfuss_2020d} may lead to the fast computation algorithm of the Kalman filter, which bypasses this kind of artifact. Lastly, since we observed spatial data abundantly at $t=2,3$, the filtering means at $t=4$ were plausible, whereas those at $t=1$ were strongly affected by $\bm{\mu}_{0|0}$ due to the shortage of spatial data. Applying the MRA-lp to the Kalman smoother on the basis of \cite{Jurek_2023} will likely be a remedy.

\section{Conclusion and future studies\label{sec_conclusion_future_studies}}

This study proposed the MRF-lp, which can quickly compute the filtering mean for linear Gaussian state-space models. Furthermore, we have extended our proposed method to nonlinear and non-Gaussian state-space models on the basis of \cite{Zilber_2021} and \cite{Jurek_2022}. The block sparse structure resulting from the MRA-lp 
makes the MRF-lp the efficient filtering algorithm. The MRF-lp can be regarded as an extension of the MRA-lp to the spatio-temporal datasets and a generalization of the MRF. 
By increasing $r_{j_1,\ldots, j_m}$ and/or $r^{\prime}_{j_1,\ldots, j_m}$ adequately, the approximation accuracy of the filtering mean can be improved while maintaining the relatively small computational time. Our proposed MRF-lp has the advantage of adjusting $r_{j_1,\ldots, j_m}$ and $r^{\prime}_{j_1,\ldots, j_m}$ separately, unlike the MRF where $r_{j_1,\ldots, j_m} = r^{\prime}_{j_1,\ldots, j_m}$. 
In the simulation studies and real data analysis, the MRF-lp was generally more efficient than the MRF in terms of approximation accuracy and computational time. Moreover, the MRF-lp empirically improved the numerical stability of the MRF.

There are some issues to be tackled in the future. First, \cite{Jurek_2021, Jurek_2022} proposed the fast computation algorithm of Rao--Blackwellized particle filter \citep{Doucet_2011} to estimate parameters in state-space models (e.g., range parameter in the covariance function). Similarly, our proposed MRF-lp would be extended to the particle version. Second, it is difficult for the extended Kalman filter to capture the strongly nonlinear evolution. Based on \cite{Chakraborty_2022}, combining the MRF-lp with the unscented Kalman filter is interesting. Third, similar to \cite{Jurek_2023}, the fast smoothing algorithm by using the MRA-lp is considered. As shown in Figure \ref{fig_prediction_surface_t1} of Section \ref{subsec_real_data}, the filtering distribution degenerates when few observations are obtained at early time points. However, the smoothing distribution may provide a plausible prediction surface. Fourth, the MRF-lp has a lot of tuning parameters such as the highest resolution $M$, the partition number $J_m$ of each subregion, the number $r_{j_1,\ldots,j_m}$ of knots, and the rank $r_{j_1,\ldots,j_m}^{\prime}$ of $\Phi_{j_1,\ldots, j_m}$. A practical guideline for choosing them would be helpful. Finally, the general Vecchia approximation proposed by \cite{Katzfuss_2021a} includes some existing fast computation methods used in spatial statistics and has shown high usefulness in a variety of situations \citep[e.g.,][]{Katzfuss_2020b, Zilber_2021, Katzfuss_2022, Jurek_2022, Jurek_2023}. It is left for a future study to explore the relationship between the MRF-lp and the fast filtering method in \cite{Jurek_2022}, which adopts a hierarchical Vecchia approximation.

%
%

\appendix

\def\thesection{Appendix \Alph{section}}
\section{Derivation of Algorithms \ref{alg_approx_cov_mat} and \ref{alg_mrf_lp_linear_non-gaussian}}
\def\thesection{\Alph{section}}
\label{append_derivation_algorithms}

\subsection{Derivation of Algorithm \ref{alg_approx_cov_mat}}
\label{subappend_derivation_alg1}

To derive Algorithm \ref{alg_approx_cov_mat}, we will apply the MRA-lp, which omits the modification via the covariance tapering, to $\Sigma_0$. In accordance with \cite{Hirano_2021}, for $0 \le l \le m, \; m = 0,\ldots,M$, we define
\begin{align*}
  W_{j_1,\ldots,j_m}^l &= \Sigma_l[I_{j_1,\ldots,j_m}, K_{j_1,\ldots,j_l}],\\
  V_{j_1,\ldots,j_m}^l &= \Sigma_l[K_{j_1,\ldots,j_m}, K_{j_1,\ldots,j_l}],
\end{align*}
where
\begin{align}
  \label{eq_Sigma_l}
  \Sigma_l[a,b] &=
  \begin{cases}
  \Sigma_{l-1}[a,b] - \Sigma_{l-1}^{\prime}[a,b], \quad \bm{g}_a,\bm{g}_b \in D_{j_1,\ldots,j_l} \; (1 \le j_i \le J_i,\; i=1,\ldots,l),\\
  0, \quad \mbox{otherwise},
  \end{cases}
  \\
  \label{eq_Sigma_l_prime}
  \Sigma_{l}^{\prime}[a,b] &=
\begin{cases}
\Sigma_l[a,K_{j_1,\ldots,j_l}] \Phi_{j_1,\ldots,j_l}^{\top} \widehat{V}_{j_1,\ldots,j_l}^{{l}^{-1}} \Phi_{j_1,\ldots,j_l} \Sigma_l[b,K_{j_1,\ldots,j_l}]^{\top}, \quad \bm{g}_a,\bm{g}_b \in D_{j_1,\ldots,j_l} \\
\hspace{7cm} (1 \le j_i \le J_i,\; i=1,\ldots,l),\\
0, \quad \mbox{otherwise}.
\end{cases}
\end{align}
For $l \ge 1$, we have
\begin{align}
  W_{j_1,\ldots,j_m}^l &= \Sigma_{l-1}[I_{j_1,\ldots,j_m}, K_{j_1,\ldots,j_l}] - \Sigma_{l-1}^{\prime}[I_{j_1,\ldots,j_m}, K_{j_1,\ldots,j_l}]\\
  &=\Sigma_0[I_{j_1,\ldots,j_m}, K_{j_1,\ldots,j_l}] - \sum_{k=0}^{l-1} \Sigma_{k}^{\prime}[I_{j_1,\ldots,j_m}, K_{j_1,\ldots,j_l}]\\
  &=\Sigma_0[I_{j_1,\ldots,j_m}, K_{j_1,\ldots,j_l}] - \sum_{k=0}^{l-1} W_{j_1,\ldots,j_m}^k \Phi_{j_1,\ldots,j_k}^{\top} \widehat{V}_{j_1,\ldots,j_k}^{{k}^{-1}} \Phi_{j_1,\ldots,j_k} V_{j_1,\ldots,j_l}^{{k}^{\top}}, \label{eq_W} \\
  V_{j_1,\ldots,j_m}^l &= \Sigma_{l-1}[K_{j_1,\ldots,j_m}, K_{j_1,\ldots,j_l}] - \Sigma_{l-1}^{\prime}[K_{j_1,\ldots,j_m}, K_{j_1,\ldots,j_l}]\\
  &=\Sigma_0[K_{j_1,\ldots,j_m}, K_{j_1,\ldots,j_l}] - \sum_{k=0}^{l-1} \Sigma_{k}^{\prime}[K_{j_1,\ldots,j_m}, K_{j_1,\ldots,j_l}]\\
  &=\Sigma_0[K_{j_1,\ldots,j_m}, K_{j_1,\ldots,j_l}] - \sum_{k=0}^{l-1} V_{j_1,\ldots,j_m}^k \Phi_{j_1,\ldots,j_k}^{\top} \widehat{V}_{j_1,\ldots,j_k}^{{k}^{-1}} \Phi_{j_1,\ldots,j_k} V_{j_1,\ldots,j_l}^{{k}^{\top}}. \label{eq_V}
  \end{align}

From the definitions of $W_{j_1,\ldots,j_m}^l$ and 
$V_{j_1,\ldots,j_m}^l$, we can obtain $W_{j_1,\ldots,j_{m+i}}^m$ ($i = 1,\ldots,M-m$) and $V_{j_1,\ldots,j_m}^m$ by restricting the rows of $W_{j_1,\ldots,j_m}^m$ to $I_{j_1,\ldots,j_{m+i}}$ and $K_{j_1,\ldots,j_m}$, respectively. 
By a similar procedure for $W_{j_1,\ldots,j_{m+i}}^m$, we can obtain $V_{j_1,\ldots,j_{m+i}}^m$ ($i = 1,\ldots,M-m$). 
Hence, for $l < m \le M$, 
we can have $W_{j_1,\ldots,j_m}^l$ and $V_{j_1,\ldots,j_m}^l$ without using \eqref{eq_W} and \eqref{eq_V}.

Finally, the approximation of $\Sigma_0 = \Sigma_0[I_0,I_0]$ 
by the MRA-lp is given by
\begin{align*}
  \sum_{m=0}^{M} \Sigma_{m}^{\prime}[I_0,I_0] 
  =& W_0^0 \Phi_0^{\top} \widehat{V}_{0}^{0^{-1}} \Phi_0 W_0^{{0}^{\top}} 
  +  
  \begin{pmatrix}
    W_1^1 \Phi_1^{\top} \widehat{V}_{1}^{1^{-1}} \Phi_1 W_1^{{1}^{\top}}  & & O \\
     & \ddots & \\
     O & & W_{J_1}^1 \Phi_{J_1}^{\top} \widehat{V}_{J_1}^{1^{-1}} \Phi_{J_1} W_{J_1}^{{1}^{\top}}
  \end{pmatrix}
\\
  + \cdots &+
  \begin{pmatrix}
    W_{1,\ldots,1}^M \Phi_{1,\ldots,1}^{\top} \widehat{V}_{1,\ldots,1}^{{M}^{-1}} \Phi_{1,\ldots,1} W_{1,\ldots,1}^{{M}^{\top}}  & & O \\
     & \ddots & \\
     O & & W_{J_1,\ldots,J_M}^{M} \Phi_{J_1,\ldots,J_M}^{\top} \widehat{V}_{J_1,\ldots,J_M}^{{M}^{-1}} \Phi_{J_1,\ldots,J_M} W_{J_1,\ldots,J_M}^{{M}^{\top}}
  \end{pmatrix}
  \\
  =& B_0 B_0^{\top} + 
  \begin{pmatrix}
    B_1 B_1^{\top}  & & O \\
     & \ddots & \\
     O & & B_{J_1} B_{J_1}^{\top}
  \end{pmatrix}
  + \cdots +
  \begin{pmatrix}
    B_{1,\ldots,1} B_{1,\ldots,1}^{\top}  & & O \\
     & \ddots & \\
     O & & B_{J_1,\ldots,J_M} B_{J_1,\ldots,J_M}^{\top}
  \end{pmatrix}
  \\
  =& B^0 B^{0^{\top}} + B^1 B^{1^{\top}} + \cdots + B^M B^{M^{\top}} \\
  =& B B^{\top}.
\end{align*}

\subsection{Derivation of Algorithm \ref{alg_mrf_lp_linear_non-gaussian}}
\label{subappend_mrf_lp_linear_non-gaussian}

From Theorem 18.2.8 of \cite{Harville_1997}, also known as the Sherman--Morrison--Woodbury formula, we obtain
\begin{align}
  W_{\bm{x}_t}^{-1} &= \left( \Sigma_{t|t-1}^{-1} + D_{\bm{x}_t} \right)^{-1} \\
&=\Sigma_{t|t-1}-\Sigma_{t|t-1} D_{\bm{x}_t}^{\frac{1}{2}} \left( \bm{\mathrm{I}}_{n_\mathcal{G}} + D_{\bm{x}_t}^{\frac{1}{2}} \Sigma_{t|t-1} D_{\bm{x}_t}^{\frac{1}{2}} \right)^{-1} D_{\bm{x}_t}^{\frac{1}{2}} \Sigma_{t|t-1}. \label{eq_w_xt_inv}
\end{align}
Next, by applying Algorithm \ref{alg_approx_cov_mat} to $\Sigma_{t|t-1}$, $\Sigma_{t|t-1}$ is approximated by $B_{t|t-1} B_{t|t-1}^{\top}$. By substituting $B_{t|t-1} B_{t|t-1}^{\top}$ into $\Sigma_{t|t-1}$ in \eqref{eq_w_xt_inv}, 
\begin{align}
  W_{\bm{x}_t}^{-1} =& B_{t|t-1} \left\{ \bm{\mathrm{I}}_{N^{\prime}} - B_{t|t-1}^{\top} D_{\bm{x}_t}^{\frac{1}{2}} \left( \bm{\mathrm{I}}_{n_\mathcal{G}} + D_{\bm{x}_t}^{\frac{1}{2}} B_{t|t-1} B_{t|t-1}^{\top} D_{\bm{x}_t}^{\frac{1}{2}} \right)^{-1} D_{\bm{x}_t}^{\frac{1}{2}} B_{t|t-1} \right\} B_{t|t-1}^{\top} \\
=& B_{t|t-1} \left( \bm{\mathrm{I}}_{N^{\prime}} + B_{t|t-1}^{\top} D_{\bm{x}_t}^{\frac{1}{2}} \bm{\mathrm{I}}_{n_\mathcal{G}} D_{\bm{x}_t}^{\frac{1}{2}} B_{t|t-1}  \right)^{-1}  B_{t|t-1}^{\top} \\
=& B_{t|t-1} \Lambda_{\bm{x}_{t}}^{-1} B_{t|t-1}^{\top}\\
=& B_{t|t-1} \left( L_{\bm{x}_{t}} L_{\bm{x}_{t}}^{\top} \right)^{-1} B_{t|t-1}^{\top} \\
=& B_{\bm{x}_{t}} B_{\bm{x}_{t}}^{\top}, \label{eq_w_xt_inv_approx}
\end{align}
where the second equality holds by using Theorem 18.2.8 of \cite{Harville_1997}. 
Thus, it follows from \eqref{eq_newton-raphson_update} and \eqref{eq_w_xt_inv_approx} that
\begin{align}
  h(\bm{x}_t) = \bm{\mu}_{t|t-1} + B_{\bm{x}_{t}} B_{\bm{x}_{t}}^{\top} \left\{ D_{\bm{x}_t} \left( \bm{x}_t - \bm{\mu}_{t|t-1} \right) + \bm{u}_{\bm{x}_t} \right\}.
\end{align}

\def\thesection{Appendix \Alph{section}}
\section{Proofs of propositions}
\def\thesection{\Alph{section}}
\label{append_proofs}

\begin{proof}[Proof of Proposition \ref{prop_mra-lp_inv}]
  We will show the assertion by mathematical induction in the same way as Propositions 1 (a) and (b) of Hirano (2021). 
  We define the random vector $\bm{Z}_0(I_0) = (Z_0(\bm{g}_1),\ldots,Z_0(\bm{g}_{n_\mathcal{G}}))^{\top} \sim \mathcal{N} \left(\bm{0}, \Sigma_0[I_0,I_0]\right)$. $\Sigma_0[I_0,I_0]$ is positive definite, so that $V_0^0 = \Sigma_0[K_0,K_0]$ is also positive definite. Since $\Phi_0$ is the full row-rank matrix, $\widehat{V}_{0}^{0}=\Phi_0 V_0^0 \Phi_0^{\top}$ is positive definite from Proposition 4.1 (a) of \cite{Hirano_2017a}. Hereafter, for a generic random vector $\bm{X}(I_0)=(X(\bm{g}_1),\ldots,X(\bm{g}_{n_\mathcal{G}}))^{\top}$ and a set $A \subset I_0$, $\bm{X}(A)$ is a subvector of $\bm{X}(I_0)$ obtained by restricting $\bm{X}(I_0)$ to the elements on grid points corresponding to $A$. Now, define the random vector
  \begin{align*}
    \bm{\tau}_0(I_0) &= (\tau_0(\bm{g}_1),\ldots,\tau_0(\bm{g}_{n_\mathcal{G}}))^{\top}\\
    &=E[\bm{Z}_0(I_0) | \Phi_0 \bm{Z}_0(K_0)] \\
    &=\Sigma_0[I_0,K_0] \Phi_0^{\top} \widehat{V}_{0}^{{0}^{-1}} \Phi_0 \bm{Z}_0(K_0).
  \end{align*}
As a consequence, we have
\begin{align*}
\Sigma_1^{\prime \prime}[I_0,I_0] &= \Sigma_0[I_0,I_0] - \mbox{Var}(\bm{\tau}_0(I_0))\\
&=\mbox{Var}(\bm{Z}_0(I_0)) - \mbox{Var}(E[\bm{Z}_0(I_0) | \Phi_0 \bm{Z}_0(K_0)])\\
&=\mbox{Var}(\bm{Z}_0(I_0) | \Phi_0 \bm{Z}_0(K_0)) \ge 0.
\end{align*}
Since $\Sigma_{0}^{\prime}[I_0,I_0]=\mbox{Var}(\bm{\tau}_0(I_0))$, from \eqref{eq_Sigma_l}, 
we can write
\begin{align*}
  \Sigma_1[I_0,I_0] =
  \begin{pmatrix}
    \Sigma_1^{\prime \prime}[I_1,I_1] & & O \\
     & \ddots & \\
     O & & \Sigma_1^{\prime\prime}[I_{J_1},I_{J_1}]
    \end{pmatrix}.  
\end{align*}

Thus, $\Sigma_1[I_0,I_0]$ is positive semidefinite, and $V_{j_1}^1 = \Sigma_1[K_{j_1},K_{j_1}]$ is also positive semidefinite. 
Combining 
a part of Proposition 5.4 of \cite{Puntanen_2011} 
with the assumption of Proposition \ref{prop_mra-lp_inv} 
yields positive definiteness of $\widehat{V}_{j_1}^{1}$. The result holds for $m=1$.

Next, assume that $\Sigma_l[I_0,I_0]$ is positive semidefinite for $m=l$ ($1 \le l \le M-1$). Consider $m=l+1$. If we define $\bm{Z}_l(I_0) = (Z_l(\bm{g}_1),\ldots,Z_l(\bm{g}_{n_\mathcal{G}}))^{\top} \sim \mathcal{N} \left(\bm{0},\Sigma_l[I_0,I_0]\right)$, then we have
\begin{align*}
  \begin{pmatrix}
    \bm{Z}_l(I_0) \\
     \Phi^{l} \bm{Z}_l (K^l)
    \end{pmatrix}
  \sim
  \mathcal{N} \left( \bm{0},
  \begin{pmatrix}
    \Sigma_l[I_0,I_0] & \Sigma_l[I_0,K^l] \Phi^{l^{\top}} \\
     \Phi^l \Sigma_l[K^l,I_0] & \Phi^l \Sigma_l[K^l,K^l] \Phi^{l^{\top}}
    \end{pmatrix}
  \right),
\end{align*}
where 
$K^l = \bigcup_{1 \le j_1 \le J_1, \ldots, 1 \le j_l \le J_l} K_{j_1,\ldots, j_l}$ 
and
\begin{align*}
  \Phi^l = 
\begin{pmatrix}
  \Phi_{1,\ldots, 1} & & O \\
   & \ddots & \\
   O & & \Phi_{J_1,\ldots, J_l}
  \end{pmatrix}.
\end{align*}
In the same way as the argument of $m=1$, $\widehat{V}_{j_1,\ldots,j_l}^{l} = \Phi_{j_1,\ldots,j_l} V_{j_1,\ldots,j_l}^l \Phi_{j_1,\ldots,j_l}^{\top}$ is positive definite. Consequently, $\Phi^l \Sigma_l[K^l, K^l] \Phi^{l^{\top}}$ is nonsingular. Then, we can define
\begin{align*}
  \bm{\tau}_l(I_0) &= (\tau_l(\bm{g}_1),\ldots,\tau_l(\bm{g}_{n_\mathcal{G}}))^{\top}\\
  &=E[\bm{Z}_l(I_0) | \Phi^l \bm{Z}_l(K^l)] \\
  &=\Sigma_l[I_0,K^l] \Phi^{l^{\top}} \left(\Phi^l \Sigma_l[K^l, K^l] \Phi^{l^{\top}} \right)^{-1} \Phi^l \bm{Z}_l(K^l).
\end{align*}
Moreover, 
\begin{align*}
  \Sigma_{l+1}^{\prime\prime}[I_0,I_0] &= \Sigma_l[I_0,I_0] - \mbox{Var}(\bm{\tau}_l(I_0))\\
  &=\mbox{Var}(\bm{Z}_l(I_0)) - \mbox{Var}(E[\bm{Z}_l(I_0) | \Phi^l \bm{Z}_l(K^l)])\\
  &=\mbox{Var}(\bm{Z}_l(I_0) | \Phi^l \bm{Z}_l(K^l)) \ge 0.
  \end{align*}
  Since $\Sigma_{l}^{\prime}[I_0,I_0]=\mbox{Var}(\bm{\tau}_l(I_0))$, from \eqref{eq_Sigma_l}, 
  it follows that
  \begin{align*}
    \Sigma_{l+1}[I_0,I_0] =
    \begin{pmatrix}
      \Sigma_{l+1}^{\prime\prime}[I_{1,\ldots,1},I_{1,\ldots,1}] & & O \\
       & \ddots & \\
       O & & \Sigma_{l+1}^{\prime\prime}[I_{J_1,\ldots,J_{l+1}},I_{J_1,\ldots,J_{l+1}}]
      \end{pmatrix}.  
  \end{align*}
  Therefore, $\Sigma_{l+1}[I_0,I_0]$ is positive semidefinite. The result holds for $m=l+1$. 
  
  Since $\Sigma_m[I_0, I_0]$ ($m=1,\ldots, M$) is positive semidefinite by mathematical induction, $\widehat{V}_{j_1,\ldots,j_m}^{m}=\Phi_{j_1,\ldots,j_m} V_{j_1,\ldots,j_m}^m \Phi_{j_1,\ldots,j_m}^{\top}$ ($m=1,\ldots, M$) is positive definite from a part of Proposition 5.4 of \cite{Puntanen_2011} and the assumption of Proposition \ref{prop_mra-lp_inv}. 
\end{proof}

\begin{proof}[Proof of Proposition \ref{prop_mra-lp_B}]
  The proof is the same as that of Proposition 1 in \cite{Jurek_2021}.
\end{proof}

\begin{proof}[Proof of Proposition \ref{prop_property_sparse_mrf_lp}]

  It is straightforward to prove Proposition \ref{prop_property_sparse_mrf_lp} (a) by using Propositions 2 and 3, Part 1 in \cite{Jurek_2021}. Similarly, the proofs of Propositions \ref{prop_property_sparse_mrf_lp} (b) and (c) are the same as those of Proposition 3, Parts 2 and 3 in \cite{Jurek_2021}, respectively.
\end{proof}

\begin{proof}[Proof of Proposition \ref{prop_property_L_L_inv}]

It follows from Proposition 4 in \cite{Jurek_2021} that the number of nonzero elements in each column of $L_t$ is $\mathcal{O}(N)$. Furthermore, by using Proposition \ref{prop_property_sparse_mrf_lp} (b), the same is true for $L_t^{-1}$.
\end{proof}

\begin{proof}[Proof of Proposition \ref{prop_property_sparse_mrf_lp_non-Gaussian}]

  The proof of the statement for the block sparse structure of $\Lambda_{\bm{x}_{t}^{(l)}}$ is based on that of Proposition 2 in \cite{Jurek_2021}. 
  $D_{\bm{x}_{t}^{(l)}}^{1/2} B_{t|t-1}$ can be viewed as a matrix with the same block sparse structure as that of $B_{t|t-1}$. Since $B_{t|t-1}$ is the output of Algorithm \ref{alg_approx_cov_mat}, we have $B_{t|t-1} = (B_{t|t-1}^M \; B_{t|t-1}^{M-1} \; \cdots \; B_{t|t-1}^0)$ where $B_{t|t-1}^m$ ($m=0,\ldots, M$) is the block diagonal matrix whose diagonal element is 
  a matrix of size $|I_{j_1,\ldots,j_m}| \times r_m^{\prime}$ ($1 \le j_i \le J, \; i = 1, \ldots, m$). Consequently, $B_{t|t-1}^{\top} D_{\bm{x}_{t}^{(l)}} B_{t|t-1}$ is a block matrix consisting of $(M+1) \times (M+1)$ blocks, and the block structure of the $(M+1-p, M+1-q)$th block is the same as that of $B_{t|t-1}^{p^{\top}} B_{t|t-1}^{q}$ ($p,q = 0,\ldots,M$). 
  
  In the case of $p \ge q$, we obtain $|I_{j_1,\ldots,j_q}| = \sum_{1 \le j_{q+1} \le J, \ldots , 1 \le j_{p} \le J} |I_{j_1,\ldots,j_p}|$ from $I_{j_1,\ldots,j_q} = \bigcup_{1 \le j_{q+1} \le J, \ldots , 1 \le j_{p} \le J} I_{j_1,\ldots,j_p}$. As a result, we can regard $B_{t|t-1}^{p}$ as the block diagonal matrix with blocks of size $|I_{j_1,\ldots,j_q}| \times J^{p-q} r_p^{\prime}$. Therefore, $B_{t|t-1}^{p^{\top}} B_{t|t-1}^{q}$ is also the block diagonal matrix of size $J^p r_p^{\prime} \times J^q r_q^{\prime}$ with blocks of size $J^{p-q} r_p^{\prime} \times r_q^{\prime}$. Since $\Lambda_{\bm{x}_{t}^{(l)}} = \bm{\mathrm{I}}_{N^{\prime}} + B_{t|t-1}^{\top} D_{\bm{x}_{t}^{(l)}} B_{t|t-1}$, 
  $\Lambda_{\bm{x}_{t}^{(l)}}$ has the same block sparse structure as that of $\Lambda_t$ in Proposition \ref{prop_property_sparse_mrf_lp} (a). 
  This leads to the statement corresponding to Proposition \ref{prop_property_sparse_mrf_lp} (b) 
  for $\Lambda_{\bm{x}_{t}^{(l)}}$, $L_{\bm{x}_{t}^{(l)}}$, and $L_{\bm{x}_{t}^{(l)}}^{-1}$. Thus, from 
  $B_{\bm{x}_{t}^{(l)}} = B_{t|t-1} L_{\bm{x}_{t}^{(l)}}^{{-1}^{\top}}$, we can prove the statement corresponding to Proposition \ref{prop_property_sparse_mrf_lp} (c) for $B_{\bm{x}_{t}^{(l)}}$ instead of $B_{t|t}$. Finally, from the statement corresponding to Proposition \ref{prop_property_sparse_mrf_lp} (b) for $\Lambda_{\bm{x}_{t}^{(l)}}$, $L_{\bm{x}_{t}^{(l)}}$, and $L_{\bm{x}_{t}^{(l)}}^{-1}$, the number of nonzero elements in each column of $L_{\bm{x}_{t}^{(l)}}$ and $L_{\bm{x}_{t}^{(l)}}^{-1}$ is $\mathcal{O}(N)$.
\end{proof}

\def\thesection{Appendix \Alph{section}}
\section{Algorithm of the EnKF for nonlinear and Gaussian state-space models}
\def\thesection{\Alph{section}}
\label{append_comparison_ensemble_kf}

Based on \cite{Katzfuss_2016, Katzfuss_2020}, we used the following algorithm as the EnKF in the nonlinear and Gaussian case.

\begin{algorithm}[Ensemble Kalman filter (EnKF) for nonlinear and Gaussian state-space models]
  \label{alg_enkf}
  
  Given $M^{\prime} > 0$ and $C_t$, find $\bm{\mu}_{t|t}$ ($t=1,\ldots,T$).
  
  \bigskip
  \noindent
  \textit{Step} 1. Generate $\bm{x}_{0|0}^{(i)}$ ($i = 1,\ldots,M^{\prime}$) from $\mathcal{N}_{n_\mathcal{G}} \left( \bm{\mu}_{0|0}, \Sigma_{0|0}  \right)$. Set $t=1$.
  
  \noindent
  \textit{Step} 2. Generate $\bm{w}_{t}^{(i)}$ ($i = 1,\ldots,M^{\prime}$) from $\mathcal{N}_{n_\mathcal{G}} \left( \bm{0}, Q_t \right)$. Calculate $\bm{x}_{t|t-1}^{(i)} = \mathcal{A}_t \left(\bm{x}_{t-1|t-1}^{(i)} \right) + \bm{w}_{t}^{(i)}$ ($i = 1,\ldots,M^{\prime}$).
  
  \noindent
  \textit{Step} 3. Generate $\bm{v}_{t}^{(i)}$ ($i = 1,\ldots,M^{\prime}$) from $\mathcal{N}_{n_t} \left( \bm{0}, R_t \right)$. Calculate $\widehat{K}_t = C_t H_t^{\top} \left( H_t C_t H_t^{\top} + R_t \right)^{-1}$. For $i=1,\ldots,M^{\prime}$, we obtain
  \begin{align*}
    \bm{x}_{t|t}^{(i)} = \bm{x}_{t|t-1}^{(i)} + \widehat{K}_t \left( \bm{y}_t - H_t \bm{x}_{t|t-1}^{(i)} + \bm{v}_{t}^{(i)} \right).
  \end{align*}
Finally, we calculate $\sum_{i=1}^{M^{\prime}} \bm{x}_{t|t}^{(i)}/M^{\prime}$ as $\bm{\mu}_{t|t}$. Set $t=t+1$. If $t \le T$, go to Step 2. Otherwise, go to Step 4.
  
  \noindent
  \textit{Step} 4. Output $\bm{\mu}_{t|t}$ ($t = 1,\ldots,T$).
  
  \end{algorithm}

%
%
\section*{Acknowledgments}
The authors would like to thank Professor Marcin Jurek and Professor Matthias Katzfuss for providing the TPW data. We also acknowledge the helpful comments and suggestions from the 
editor and two anonymous referees that refined the manuscript. This work was supported by JSPS KAKENHI Grant Number JP21K13273.

\section*{Disclosure statement}
The authors report there are no competing interests to declare.

%
%

\bibliographystyle{apalike}
\bibliography{20250224_ref}   


\newpage

\setcounter{page}{1}
\setcounter{equation}{0}
\renewcommand{\theequation}{S\arabic{equation}}
\setcounter{section}{0}
\renewcommand{\thesection}{S\arabic{section}}
\setcounter{lemma}{0}
\renewcommand{\thelemma}{S\arabic{lem}}
\setcounter{theorem}{0}
\renewcommand{\thetheorem}{S\arabic{thm}}
\setcounter{table}{0}
\renewcommand{\thetable}{S\arabic{table}}
\setcounter{figure}{0}
\renewcommand{\thefigure}{S\arabic{figure}}

\begin{center}
  {\LARGE {\bf Supplementary material to ``Multi-resolution filters via linear projection for large spatio-temporal datasets''}}
\end{center}

\section{Derivation of $A_t$ from the advection-diffusion equation}
\label{suppl_advection_diffusion}

For $\bm{s} = (s_1,s_2)^{\top} \in D_0
 \subset \mathbb{R}^2$ and $t \in [0,T]$, we consider the following advection-diffusion equation
\begin{align}
  \frac{\partial}{\partial t} x(\bm{s},t) = \alpha \left( \frac{\partial}{\partial s_1} x(\bm{s},t) + \frac{\partial}{\partial s_2} x(\bm{s},t) \right) + \beta \left( \frac{\partial^2}{\partial s_1^2} x(\bm{s},t)+\frac{\partial^2}{\partial s_2^2} x(\bm{s},t) \right). \label{eq_advection_diffusion_eq}
\end{align}
The first and second terms of \eqref{eq_advection_diffusion_eq} control the advection and the diffusion, respectively. Based on \cite{Xu_2007}, \cite{Stroud_2010}, and \cite{Jurek_2021}, we discretize \eqref{eq_advection_diffusion_eq} over a regular grid and approximate the time and spatial derivatives in \eqref{eq_advection_diffusion_eq} by the first-order forward differences in time and the centered differences in space, respectively. Thus, for positive $\Delta t$, $\Delta s_1$, and $\Delta s_2$, 
\begin{align}
  \frac{\partial}{\partial t} x(\bm{s},t) &\approx 
  \frac{x(\bm{s},t+\Delta t)-x(\bm{s},t)}{\Delta t}, \label{eq_forward_diff_time} \\
  \frac{\partial}{\partial s_1} x(\bm{s},t) &\approx 
  \frac{x(\bm{s}+(\Delta s_1,0)^{\top},t)-x(\bm{s}-(\Delta s_1,0)^{\top},t)}{2 \Delta s_1}, \label{eq_centered_diff_space_s1} \noeqref{eq_centered_diff_space_s1} \\
  \frac{\partial}{\partial s_2} x(\bm{s},t) &\approx 
  \frac{x(\bm{s}+(0, \Delta s_2)^{\top},t)-x(\bm{s}-(0, \Delta s_2)^{\top},t)}{2 \Delta s_2}, \label{eq_centered_diff_space_s2} \noeqref{eq_centered_diff_space_s2} \\
  \frac{\partial^2}{\partial s_1^2} x(\bm{s},t) &\approx 
  \frac{x(\bm{s}-(\Delta s_1,0)^{\top},t) -2 x(\bm{s},t) + x(\bm{s}+(\Delta s_1,0)^{\top},t)}{ (\Delta s_1)^2}, \label{eq_centered_twice_diff_space_s1} \noeqref{eq_centered_twice_diff_space_s1} \\
  \frac{\partial^2}{\partial s_2^2} x(\bm{s},t) &\approx 
  \frac{x(\bm{s}-(0, \Delta s_2)^{\top},t) -2 x(\bm{s},t) + x(\bm{s}+(0, \Delta s_2)^{\top},t)}{ (\Delta s_2)^2}. \label{eq_centered_twice_diff_space_s2}
\end{align}
We take $\Delta t = 1$. By substituting \eqref{eq_forward_diff_time}--\eqref{eq_centered_twice_diff_space_s2} into \eqref{eq_advection_diffusion_eq}, we obtain
\begin{align}
    x(\bm{s},t+1) = &c_0 x(\bm{s},t) + c_{-1} x(\bm{s}-(\Delta s_1,0)^{\top},t) + c_1 x(\bm{s}+(\Delta s_1,0)^{\top},t) \\
  &+ c_{-2} x(\bm{s}-(0, \Delta s_2)^{\top},t) + c_2 x(\bm{s}+(0, \Delta s_2)^{\top},t), \label{eq_discretize_advection_diffusion_eq}
\end{align}
where 
\begin{align}
c_0 &= 1-\frac{2\beta}{(\Delta s_1)^2}-\frac{2\beta}{(\Delta s_2)^2},\\
c_{-1} &= \frac{\beta}{(\Delta s_1)^2}-\frac{\alpha}{2 \Delta s_1},\\
c_1 &= \frac{\beta}{(\Delta s_1)^2}+\frac{\alpha}{2 \Delta s_1},\\
c_{-2} &= \frac{\beta}{(\Delta s_2)^2}-\frac{\alpha}{2 \Delta s_2},\\
c_2 &= \frac{\beta}{(\Delta s_2)^2}+\frac{\alpha}{2 \Delta s_2}.
\end{align}

If we have $( x(\bm{g}_1,t), \ldots, x(\bm{g}_{n_{\mathcal{G}}},t))^{\top}$ 
on a grid $\mathcal{G} = \{ \bm{g}_1, \ldots, \bm{g}_{n_\mathcal{G}} \}$, and $\Delta s_1$ and $\Delta s_2$ correspond to the interval between adjacent grid points along each dimension, we can obtain $A_t$ from \eqref{eq_discretize_advection_diffusion_eq}. 
If there are no adjacent variables for $x(\bm{s},t)$ at the edge of the grid points $\mathcal{G}$, we take the corresponding coefficients in \eqref{eq_discretize_advection_diffusion_eq} to be equal to zero. Finally, each row of $A_t$ includes at most five nonzero elements given by $c_0$, $c_{-1}$, $c_1$, $c_{-2}$, $c_2$. We substitute this $A_t$ into linear state-space models.

In the simulation studies, except for the fifth one in Section \ref{subsec_simulation_study}, we set $\alpha = 0.01$, $\beta=0.0002$, $\Delta s_1 = \Delta s_2 = 1/35$ for $A_t$. In this case, the coordinates of the grid points were $(1/35,1/35), (1/35,2/35),\ldots, (34/35,34/35)$. In Figure \ref{fig_simulation_examples_advection_diffusion}, we provide examples of sample realizations of \eqref{eq_state_equation} in the baseline of Section \ref{subsec_simulation_study}.
\begin{figure}[t]
  \centering
   \subfigure[$t=1$]{
    \includegraphics[width = 6.1cm,pagebox=artbox,clip]{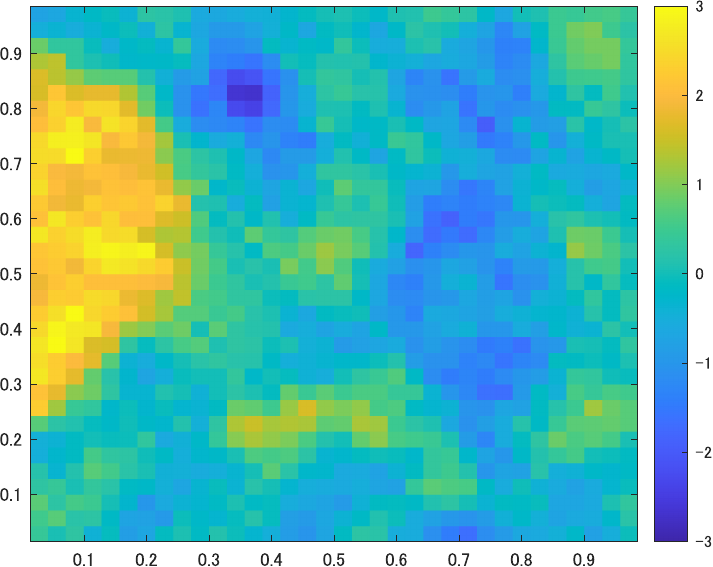}}
   \subfigure[$t=7$]{ 
    \includegraphics[width = 6.1cm,pagebox=artbox,clip]{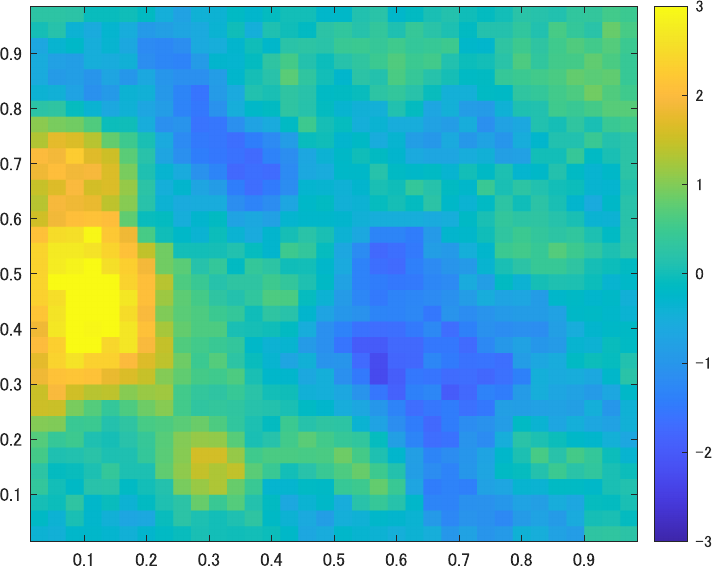}}
   \subfigure[$t=13$]{
    \includegraphics[width = 6.1cm,pagebox=artbox,clip]{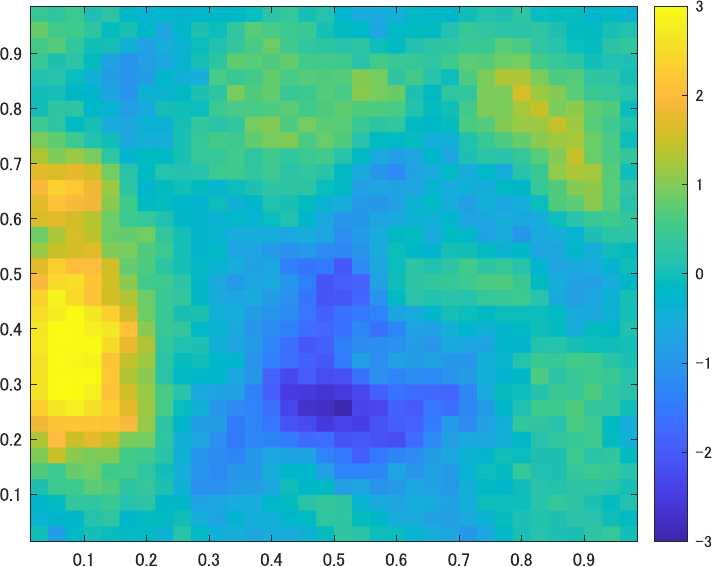}}
   \subfigure[$t=20$]{
    \includegraphics[width = 6.1cm,pagebox=artbox,clip]{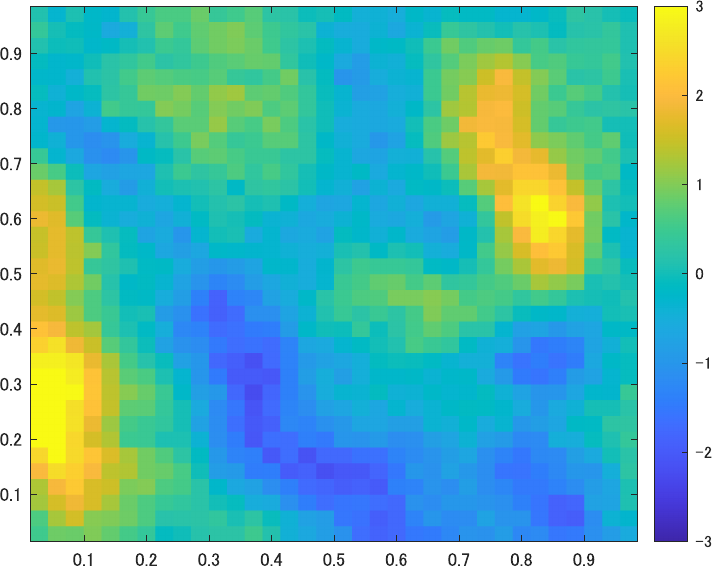}}
    \caption{Examples of sample realizations of \eqref{eq_state_equation} by using the discretization of the advection-diffusion equation.}
  \label{fig_simulation_examples_advection_diffusion} 
\end{figure}
In the real data analysis of Section \ref{subsec_real_data}, we employ $\alpha = 0$, $\beta=0.000003$, $\Delta s_1 = \Delta s_2 = 1/101$ based on \cite{Jurek_2023}.

\section{Derivation of $\mathcal{A}_t$ from  \cite{Lorenz_2005}}
\label{suppl_lorenz_2005}

To derive $\mathcal{A}_t$ in nonlinear state-space models, based on \cite{Jurek_2022}, we consider the system of partial differential equations in \cite{Lorenz_2005}, which is given by
\begin{align}
  \frac{\partial}{\partial t} x(s_i,t) = x(s_{i-1},t) \left( x(s_{i+1},t) - x(s_{i-2},t) \right) - x(s_i,t) +F,
 \label{eq_lorenz_model_I}
\end{align}
where $s_i$ ($i=1,\ldots,n_{\mathcal{G}}$) is a regular grid on a unit circle, $F$ means external forcing, and
\begin{align}
  x(s_{i},t) =
    \begin{cases}
      x(s_{i+n_{\mathcal{G}}},t), \quad  i<1,\\
      x(s_{i-n_{\mathcal{G}}},t), \quad  i > n_{\mathcal{G}}.
    \end{cases}
\end{align}
\eqref{eq_lorenz_model_I} models the scalar atmospheric quantity on a latitude circle.

We solve \eqref{eq_lorenz_model_I} numerically by using the fourth-order Runge--Kutta method and obtain the explicit expression of $\mathcal{A}_t$, which is used in nonlinear state-space models. Additionally, we calculate the Jacobian matrix $A_t^{\prime}$ from $\mathcal{A}_t$. It is straightforward to derive $A_t^{\prime}$, but its expression is somewhat complicated. In the simulation studies, except for the fifth one in Section \ref{subsec_simulation_study}, we set $F=0.5$. Figure \ref{fig_simulation_nonlinear_examples_Lorenz} illustrates examples of sample realizations of \eqref{eq_system_model_nonlinear} with $\mathcal{A}_t$ used in the third and fourth simulation studies in Section \ref{subsec_simulation_study}.

\begin{figure}[t]
  \centering
   \subfigure[$t=1$]{
    \includegraphics[width = 6.1cm,pagebox=artbox,clip]{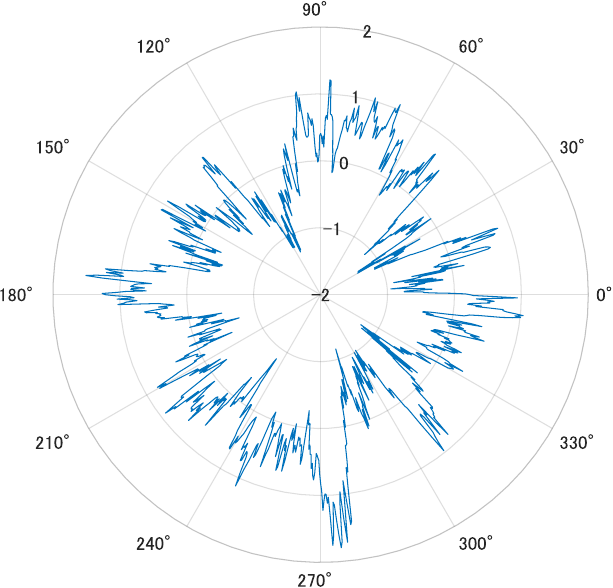}}
   \subfigure[$t=7$]{ 
    \includegraphics[width = 6.1cm,pagebox=artbox,clip]{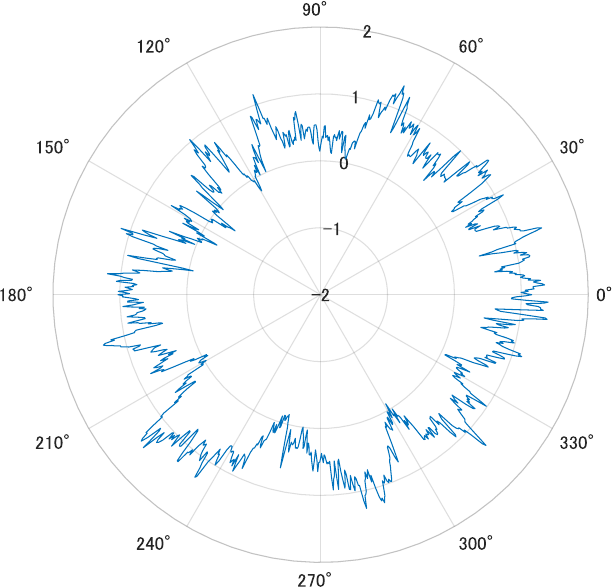}}
   \subfigure[$t=13$]{
    \includegraphics[width = 6.1cm,pagebox=artbox,clip]{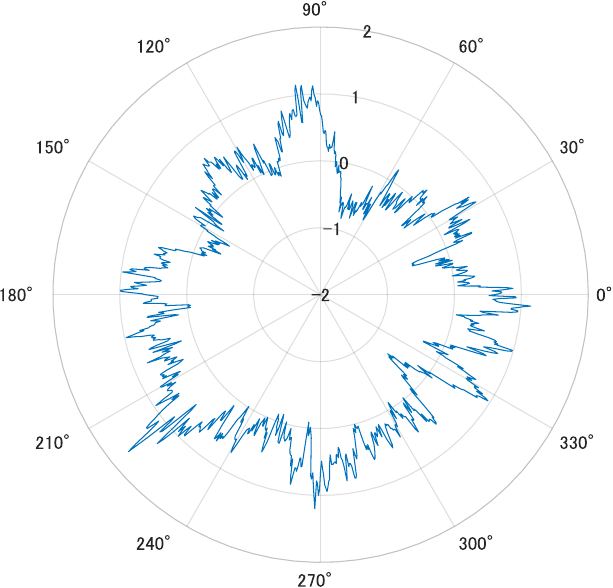}}
   \subfigure[$t=20$]{
    \includegraphics[width = 6.1cm,pagebox=artbox,clip]{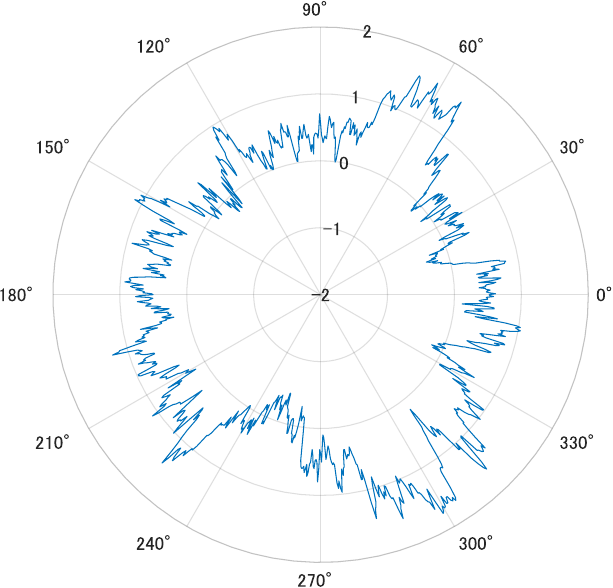}}
    \caption{
    Examples of sample realizations of \eqref{eq_system_model_nonlinear} 
    with $\mathcal{A}_t$ obtained by applying the fourth-order Runge--Kutta method to \eqref{eq_lorenz_model_I}.
    }
  \label{fig_simulation_nonlinear_examples_Lorenz} 
\end{figure}

\end{document}